\documentclass[peerreview,a4paper,12pt]{IEEEtran}

\usepackage{amsthm}
\usepackage{amsmath}
\usepackage{dsfont}
\usepackage{amsfonts,amssymb,amsmath}
\usepackage{graphicx}
\usepackage{epsfig}
\usepackage[applemac]{inputenc}
\usepackage{latexsym}

\newtheorem{mydef}{Definition}
\newtheorem{mycor}{Corollary}
\newtheorem{myprop}{Proposition}
\newtheorem{mythe}{Theorem}
\newtheorem{mylem}{Lemma}

\DeclareMathOperator*{\DMC}{DMC}
\DeclareMathOperator*{\Proj}{Proj}

\DeclareMathOperator*{\conv}{co}

\DeclareMathOperator*{\CE}{CE}

\DeclareMathOperator*{\irank}{irank}

\DeclareMathOperator*{\opt}{opt}
\DeclareMathOperator*{\ach}{ach}

\begin{document}

\sloppy

\title{On the Input-Degradedness and Input-Equivalence Between Channels}


%
%

\author{
Rajai Nasser\\
EPFL, Lausanne, Switzerland\\
rajai.nasser@epfl.ch
}




\maketitle

\begin{abstract}
A channel $W$ is said to be input-degraded from another channel $W'$ if $W$ can be simulated from $W'$ by randomization at the input. We provide a necessary and sufficient condition for a channel to be input-degraded from another one. We show that any decoder that is good for $W'$ is also good for $W$. We provide two characterizations for input-degradedness, one of which is similar to the Blackwell-Sherman-Stein theorem. We say that two channels are input-equivalent if they are input-degraded from each other. We study the topologies that can be constructed on the space of input-equivalent channels, and we investigate their properties. Moreover, we study the continuity of several channel parameters and operations under these topologies.
\end{abstract}

\section{Introduction}

The ordering of communication channels was first introduced by Shannon in \cite{ShannonDegrad}. A channel $W'$ is said to contain another channel $W$ if $W$ can be simulated from $W'$ by randomization at the input and the output using a shared randomness between the transmitter and the receiver. Shannon showed that the existence of an $(n,M,\epsilon)$ code for $W$ implies the existence of an $(n,M,\epsilon)$ code for $W'$.

Another ordering that has been well studied is the degradedness between channels. A channel $W$ is said to be degraded from another channel $W'$ if $W$ can be simulated from $W'$ by randomization at the output, or more precisely, if $W$ can be obtained from $W'$ by composing it with another channel. It is easy to see that degradedness is a special case of Shannon's ordering. One can trace the roots of the notion of degradedness to the seminal work of Blackwell in the 1950's about comparing statistical experiments \cite{blackwell1951}. Note that in the Shannon's ordering, the input and output alphabets need not be the same, whereas in the degradedness definition, we have to assume that $W$ and $W'$ share the same input alphabet $\mathcal{X}$ but they can have different output alphabets.

It is well known that if $W$ is degraded from $W'$, then for any fixed code $\mathcal{C}\subset\mathcal{X}^n$, the probability of error of the ML decoder for $\mathcal{C}$ when it is used for $W'$ is at least as good as the probability of error of the ML decoder for $\mathcal{C}$ when it is used for $W$.

In this paper, we introduce another special case of the Shannon ordering that we call \emph{input-degradedness}. A channel $W$ is said to be input-degraded from another channel $W'$ if $W$ can be simulated from $W'$ by randomization at the input. Note that $W$ and $W'$ must have the same output alphabet, but they can have different input alphabets. We say that two channels are input-equivalent if they are input-degraded from each other.

One motivation to study the input-degradedness ordering is the following: let $W$ be a fixed channel with input alphabet $\mathcal{X}$ and output alphabet $\mathcal{Y}$. Assume that after some effort, an engineer came up with a good encoder/decoder pair for $W$ in the sense that the probability of error is small. Assume also that the designed decoder is particularly desirable for some reason (e.g., it has a low computational complexity) so that we would like to use it for other channels if possible. What are the channels $W'$ for which the designed decoder also performs well in the sense that there exists a code having a low probability of error under the same decoder? We will show that a sufficient condition for the decoder to perform well for $W'$ is the input-degradedness of $W$ with respect to $W'$.

In \cite{RajDMCTop} and \cite{RajContTop}, we constructed topologies for the space of equivalent channels and studied the continuity of various channel parameters and operations under these topologies. In this paper, we show that many of the results in \cite{RajDMCTop} and \cite{RajContTop} can be replicated (with some variation) for the space of input-equivalent channels.

In Section II, we introduce the preliminaries for this paper. In Section III, we introduce and study the input-degradedness ordering. Various operational implications and characterizations of input-degradedness are provided in Section IV. The quotient topology of the space of input-equivalent channels with fixed input and output alphabets is studied in Section V. The space of input-equivalent channels with fixed output alphabet and arbitrary but finite input alphabet is defined in Section VI. A topology on this space is said to be natural if it induces the quotient topology on the subspaces of input-equivalent channels with fixed input alphabet. In Section VI, we investigate the properties of natural topologies. The finest natural topology, which we call the strong topology, is studied in Section VII. The similarity metric on the space of input-equivalent channels is introduced in Section VIII. We study the continuity of various channel parameters and operations under the strong and similarity topologies in Section IX. Finally, we show that the Borel $\sigma$-algebra is the same for all Hausdorff natural topologies.

\section{Preliminaries}

We assume that the reader is familiar with the basic concepts of general topology. The main concepts and theorems that we need can be found in the preliminaries section of \cite{RajDMCTop}.

\subsection{Measure theoretic notations}

The set of probability measures on a measurable space $(M,\Sigma)$ is denoted as $\mathcal{P}(M,\Sigma)$. For every $P_1,P_2\in\mathcal{P}(M,\Sigma)$, the \emph{total variation distance} between $P_1$ and $P_2$ is defined as:
$$\|P_1-P_2\|_{TV}=\sup_{A\in\Sigma}|P_1(A)-P_2(A)|.$$

Let $P$ be a probability measure on $(M,\Sigma)$, and let $f:M\rightarrow M'$ be a measurable mapping from $(M,\Sigma)$ to another measurable space $(M',\Sigma')$. The \emph{push-forward probability measure of $P$ by $f$} is the probability measure $f_{\#}P$ on $(M',\Sigma')$ defined as $(f_{\#}P)(A')=P(f^{-1}(A'))$ for every $A'\in\Sigma'$. If $\mathcal{A}$ is a subset of $\mathcal{P}(M,\Sigma)$, we define its push-forward by $f$ as $f_{\#}(\mathcal{A})=\{f_{\#}P:\;P\in \mathcal{A}\}$.

We denote the product of two measurable spaces $(M_1,\Sigma_1)$ and $(M_2,\Sigma_2)$ as $(M_1\times M_2,\Sigma_1\otimes\Sigma_2)$. If $P_1\in\mathcal{P}(M_1,\Sigma_1)$ and $P_2\in\mathcal{P}(M_2,\Sigma_2)$, we denote the product of $P_1$ and $P_2$ as $P_1\times P_2$. Let $\mathcal{A}_1$ and $\mathcal{A}_2$ be two subsets of $\mathcal{P}(M_1,\Sigma_1)$ and $\mathcal{P}(M_2,\Sigma_2)$ respectively. We define the tensor product of $\mathcal{A}_1$ and $\mathcal{A}_2$ as follows:
$$\mathcal{A}_1\otimes\mathcal{A}_2=\{P_1\times P_2:\; P_1\in\mathcal{A}_1,\;P_2\in\mathcal{A}_2\}\subset \mathcal{P}(M_1\times M_2,\Sigma_1\otimes\Sigma_2).$$

If $\mathcal{X}$ is a finite set, we denote the set of probability distributions on $\mathcal{X}$ as $\Delta_{\mathcal{X}}$. We always endow $\Delta_{\mathcal{X}}$ with the total variation distance and its induced topology.

\subsection{The space of channels from $\mathcal{X}$ to $\mathcal{Y}$}

Let $\DMC_{\mathcal{X},\mathcal{Y}}$ be the set of all channels having $\mathcal{X}$ as input alphabet and $\mathcal{Y}$ as output alphabet. For every $W,W'\in\DMC_{\mathcal{X},\mathcal{Y}}$, define the distance between $W$ and $W'$ as:
$$d_{\mathcal{X},\mathcal{Y}}(W,W')=\frac{1}{2} \max_{x\in\mathcal{X}}\sum_{y\in\mathcal{Y}}|W'(y|x)-W(y|x)|.$$

Throughout this paper, we always associate the space $\DMC_{\mathcal{X},\mathcal{Y}}$ with the metric distance $d_{\mathcal{X},\mathcal{Y}}$ and the metric topology $\mathcal{T}_{\mathcal{X},\mathcal{Y}}$ induced by it. It is easy to see that $\mathcal{T}_{\mathcal{X},\mathcal{Y}}$ is the same as the topology inherited from the Euclidean topology of $\mathbb{R}^{\mathcal{X}\times\mathcal{Y}}$ by relativization. It is also easy to see that the metric space $\DMC_{\mathcal{X},\mathcal{Y}}$ is compact and path-connected (see \cite{RajDMCTop}).

For every $W\in\DMC_{\mathcal{X},\mathcal{Y}}$ and every $V\in\DMC_{\mathcal{Y},\mathcal{Z}}$, define the composition $V\circ W\in\DMC_{\mathcal{X},\mathcal{Z}}$ as
$$(V\circ W)(z|x)=\sum_{y\in\mathcal{Y}}V(z|y)W(y|x),\;\;\forall x\in\mathcal{X},\;\forall z\in\mathcal{Z}.$$

For every mapping $f:\mathcal{X}\rightarrow\mathcal{Y}$, define the deterministic channel $D_f\in\DMC_{\mathcal{X},\mathcal{Y}}$ as $$D_f(y|x)=\begin{cases}1\quad&\text{if}\;y=f(x),\\0\quad&\text{otherwise}.\end{cases}$$
It is easy to see that if $f:\mathcal{X}\rightarrow \mathcal{Y}$ and $g:\mathcal{Y}\rightarrow \mathcal{Z}$, then $D_g\circ D_f=D_{g\circ f}$.

\subsection{Convex-extreme points}
Let $\mathcal{X}$ be a finite set. For every $A\subset \Delta_{\mathcal{X}}$, let $\conv(A)$ be the convex hull of $A$. We say that $p\in A$ is \emph{convex-extreme} if it is an extreme point of $\conv(A)$, i.e., for every $p_1,\ldots,p_n\in\conv(A)$ and every $\lambda_1,\ldots,\lambda_n>0$ satisfying $\displaystyle\sum_{i=1}^n \lambda_i=1$ and $\displaystyle\sum_{i=1}^n \lambda_ip_i=p$, we have $p_1=\ldots=p_n=p$. It is easy to see that if $A$ is finite, then the convex-extreme points of $A$ coincide with the extreme points of $\conv(A)$. We denote the set of convex-extreme points of $A$ as $\CE(A)$.

\subsection{The Hausdorff metric}
Let $(M,d)$ be a metric space. Let $\mathcal{K}(M)$ be the set of compact subsets of $M$. The Hausdorff metric on $\mathcal{K}(M)$ is defined as:
\begin{align*}
d_H(K_1,K_2)&=\max\left\{\sup_{x_1\in K_1} d(x_1,K_2), \sup_{x_2\in K_2} d(x_2,K_1) \right\} \\
&=\max\left\{\sup_{x_1\in K_1} \inf_{x_2\in K_2} d(x_1,x_2), \sup_{x_2\in K_2} \inf_{x_1\in K_1} d(x_2,x_1) \right\}.
\end{align*}

\subsection{Quotient topology}

Let $(T,\mathcal{U})$ be a topological space and let $R$ be an equivalence relation on $T$. The \emph{quotient topology} on $T/R$ is the finest topology that makes the projection mapping $\Proj_R$ onto the equivalence classes continuous. It is given by
$$\mathcal{U}/R=\left\{\hat{U}\subset T/R:\;\textstyle\Proj_R^{-1}(\hat{U})\in \mathcal{U}\right\}.$$

\begin{mylem}
\label{lemQuotientFunction}
Let $f:T\rightarrow S$ be a continuous mapping from $(T,\mathcal{U})$ to $(S,\mathcal{V})$. If $f(x)=f(x')$ for every $x,x'\in T$ satisfying $x R x'$, then we can define a \emph{transcendent mapping} $f:T/R\rightarrow S$ such that $f(\hat{x})=f(x')$ for any $x'\in\hat{x}$. $f$ is well defined on $T/R$ . Moreover, $f$ is a continuous mapping from $(T/R,\mathcal{U}/R)$ to $(S,\mathcal{V})$.
\end{mylem}

Let $(T,\mathcal{U})$ and $(S,\mathcal{V})$ be two topological spaces and let $R$ be an equivalence relation on $T$. Consider the equivalence relation $R'$ on $T\times S$ defined as $(x_1,y_1) R' (x_2,y_2)$ if and only if $x_1 R x_2$ and $y_1=y_2$. A natural question to ask is whether the canonical bijection between $\big((T/R)\times S,(\mathcal{U}/R)\otimes \mathcal{V} \big)$ and $\big((T\times S)/R',(\mathcal{U}\otimes \mathcal{V})/R' \big)$ is a homeomorphism. It turns out that this is not the case in general. The following theorem, which is widely used in algebraic topology, provides a sufficient condition:

\begin{mythe}
\label{theQuotientProd}
\cite{Engelking}
If $(S,\mathcal{V})$ is locally compact and Hausdorff, then the canonical bijection between $\big((T/R)\times S,(\mathcal{U}/R)\otimes \mathcal{V} \big)$ and $\big((T\times S)/R',(\mathcal{U}\otimes \mathcal{V})/R' \big)$ is a homeomorphism.
\end{mythe}

\begin{mycor}
\label{corQuotientProd}
\cite{RajContTop} Let $(T,\mathcal{U})$ and $(S,\mathcal{V})$ be two topological spaces, and let $R_T$ and $R_S$ be two equivalence relations on $T$ and $S$ respectively. Define the equivalence relation $R$ on $T\times S$ as $(x_1,y_1) R (x_2,y_2)$ if and only if $x_1 R_T x_2$ and $y_1R_S y_2$. If $(S,\mathcal{V})$ and $(T/R_T,\mathcal{U}/R_T)$ are locally compact and Hausdorff, then the canonical bijection between $\big((T/R_T)\times (S/R_S),(\mathcal{U}/R_T)\otimes (\mathcal{V}/R_S) \big)$ and $\big((T\times S)/R,(\mathcal{U}\otimes \mathcal{V})/R \big)$ is a homeomorphism.
\end{mycor}

\section{Input-degradedness and input-equivalence}

\label{secInputDegradDef}

Let $\mathcal{X},\mathcal{X}'$ and $\mathcal{Y}$ be three finite sets. Let $W\in \DMC_{\mathcal{X},\mathcal{Y}}$ and $W'\in \DMC_{\mathcal{X}',\mathcal{Y}}$. We say that $W$ is \emph{input-degraded} from $W'$ if there exists a channel $V'\in \DMC_{\mathcal{X},\mathcal{X}'}$ such that $W=W'\circ V'$. The channels $W$ and $W'$ are said to be \emph{input-equivalent} if each one is input-degraded from the other.

Let $W\in\DMC_{\mathcal{X},\mathcal{Y}}$ be a fixed channel with input alphabet $\mathcal{X}$ and output alphabet $\mathcal{Y}$. For every $x\in\mathcal{X}$, define $W_x\in\Delta_{\mathcal{Y}}$ as:
$$W_x(y)=W(y|x),\;\;\forall y\in\mathcal{Y}.$$

\begin{myprop}
\label{propCharacInpDegrad}
Let $\mathcal{X}',\mathcal{X}$ and $\mathcal{Y}$ be three finite sets. $W\in\DMC_{\mathcal{X},\mathcal{Y}}$ is input-degraded from $W'\in\DMC_{\mathcal{X}',\mathcal{Y}}$ if and only if $\conv(\{W_x:\;x\in\mathcal{X}\})\subset \conv(\{W_{x'}':\;x'\in\mathcal{X}'\})$.
\end{myprop}
\begin{proof}
Assume that $W$ is input-degraded from $W'$. There exists $V'\in\DMC_{\mathcal{X},\mathcal{X}'}$ such that $W=W'\circ V'$. For every $x\in\mathcal{X}$ and $y\in\mathcal{Y}$, we have:
$$W_x(y)=W(y|x)=\sum_{x'\in\mathcal{X}'}W'(y|x')V'(x'|x)=\sum_{x'\in\mathcal{X}'}V'(x'|x)W'_{x'}(y).$$
Therefore, $W_x=\displaystyle\sum_{x'\in\mathcal{X}'}V'(x'|x)W'_{x'}$ which means that $W_x\in \conv(\{W_{x'}':\;x'\in\mathcal{X}'\})$ for every $x\in\mathcal{X}$, hence $\conv(\{W_x:\;x\in\mathcal{X}\})\subset \conv(\{W_{x'}':\;x'\in\mathcal{X}'\})$.

Conversely, assume that $\conv(\{W_x:\;x\in\mathcal{X}\})\subset \conv(\{W_{x'}':\;x'\in\mathcal{X}'\})$ and let $x\in\mathcal{X}$. Since $W_x\in \conv(\{W_{x'}':\;x'\in\mathcal{X}'\})$, there exists a set of numbers $\alpha_{x,x'}\geq 0$ satisfying $\displaystyle\sum_{x'\in\mathcal{X}'}\alpha_{x,x'}=1$ such that $W_x=\displaystyle\sum_{x'\in\mathcal{X}'}\alpha_{x,x'}W_{x'}$. Define $V'\in\DMC_{\mathcal{X},\mathcal{X}'}$ as $V(x'|x)=\alpha_{x,x'}$ for every $x\in\mathcal{X}$ and every $x'\in\mathcal{X}'$. We have $W=W'\circ V'$ and so $W$ is input-degraded from $W'$.
\end{proof}

\vspace*{3mm}

For every channel $W\in\DMC_{\mathcal{X},\mathcal{Y}}$, we define the \emph{input-equivalence characteristic of $W$}, or simply the \emph{characteristic of $W$}, as $\CE(W):=\CE(\{W_x:\;x\in\mathcal{X}\})$. The \emph{input-rank} of $W\in\DMC_{\mathcal{X},\mathcal{Y}}$ is the size of its characteristic: $\irank(W)=|\CE(W)|$.

\begin{myprop}
\label{propCharacInputEquiv}
Let $\mathcal{X}',\mathcal{X}$ and $\mathcal{Y}$ be three finite sets. $W\in\DMC_{\mathcal{X},\mathcal{Y}}$ is input-equivalent to $W'\in\DMC_{\mathcal{X}',\mathcal{Y}}$ if and only if $\CE(W)=\CE(W')$.
\end{myprop}
\begin{proof}
It follows from Proposition \ref{propCharacInpDegrad} that $W$ is input-equivalent to $W'$ if and only if $\conv(\{W_x:\;x\in\mathcal{X}\})= \conv(\{W_{x'}':\;x'\in\mathcal{X}'\})$, which happens if and only if $\CE(W)=\CE(\conv(\{W_x:\;x\in\mathcal{X}\}))= \CE(\conv(\{W_{x'}':\;x'\in\mathcal{X}'\}))=\CE(W')$.
\end{proof}

\section{Operational implications of input-degradedness}

\subsection{Operational implication in terms of decoders}

Let $\mathcal{Y}$ be a finite set. An $(n,M)$-decoder on $\mathcal{Y}$ is a mapping $\mathcal{D}:\mathcal{Y}^n\rightarrow\mathcal{M}$, where $|\mathcal{M}|=M$. The set $\mathcal{M}$ is the \emph{message set} of  $\mathcal{D}$, $n$ is the \emph{blocklength} of $\mathcal{D}$, $M$ is the \emph{size} of $\mathcal{D}$ and $\frac{1}{n}\log|\mathcal{M}|$ is the \emph{rate} of $\mathcal{D}$ (measured in nats).

Let $W\in\DMC_{\mathcal{X},\mathcal{Y}}$ be a channel with input alphabet $\mathcal{X}$ and output alphabet $\mathcal{Y}$, and let $\mathcal{D}:\mathcal{Y}^n\rightarrow\mathcal{M}$ be a decoder on $\mathcal{Y}$. A \emph{maximum-likelihood (ML) encoder} for $\mathcal{D}$ when it is used for $W$ is any encoder $\mathcal{E}:\mathcal{M}\rightarrow\mathcal{X}^n$ satisfying
$$\sum_{\substack{y_1^n\in\mathcal{Y}^n,\\\mathcal{D}(y_1^n)=m}} \prod_{i=1}^n W(y_i|\mathcal{E}_i(m))\geq \sum_{\substack{y_1^n\in\mathcal{Y}^n,\\\mathcal{D}(y_1^n)=m}} \prod_{i=1}^n W(y_i|x_i),\;\;\forall m\in\mathcal{M},\;\forall x_1^n\in\mathcal{X}^n,$$
where $(\mathcal{E}_1(m),\ldots,\mathcal{E}_n(m))=\mathcal{E}(m)\in\mathcal{X}^n$.

It is easy to see that a maximum-likelihood encoder has the best probability of error among all encoders (assuming that the decoder $\mathcal{D}$ is used). The \emph{probability of error of $\mathcal{D}$ under ML-encoding for $W$} is given by:
$$P_{e,\mathcal{D}}(W)=1-\frac{1}{|\mathcal{M}|}\sum_{m\in\mathcal{M}} \max_{x_1^n\in\mathcal{X}^n}\Bigg\{ \sum_{\substack{y_1^n\in\mathcal{Y}^n,\\\mathcal{D}(y_1^n)=m}} \prod_{i=1}^n W(y_i|x_i)\Bigg\}.$$

\begin{myprop}
\label{propInputDegradOperational}
Let $\mathcal{X}',\mathcal{X}$ and $\mathcal{Y}$ be three finite sets. If $W\in\DMC_{\mathcal{X},\mathcal{Y}}$ is input-degraded from $W'\in\DMC_{\mathcal{X}',\mathcal{Y}}$, then $P_{e,\mathcal{D}}(W')\leq P_{e,\mathcal{D}}(W)$ for every decoder $\mathcal{D}$ on $\mathcal{Y}$. Moreover, if $W$ and $W'$ are input-equivalent, then $P_{e,\mathcal{D}}(W)= P_{e,\mathcal{D}}(W')$ for every decoder $\mathcal{D}$ on $\mathcal{Y}$.
\end{myprop}
\begin{proof}
Assume that $W\in\DMC_{\mathcal{X},\mathcal{Y}}$ is input-degraded from $W'\in\DMC_{\mathcal{X}',\mathcal{Y}}$. Let $V'\in\DMC_{\mathcal{X},\mathcal{X}'}$ be such that $W=W'\circ V'$.

Fix an $(n,M)$ decoder $\mathcal{D}$ on $\mathcal{Y}$ and let $\mathcal{M}$ be its message set. We have:
\begin{align*}
1-P_{e,\mathcal{D}}(W)&=\frac{1}{|\mathcal{M}|}\sum_{m\in\mathcal{M}} \max_{x_1^n\in\mathcal{X}^n}\Bigg\{ \sum_{\substack{y_1^n\in\mathcal{Y}^n,\\\mathcal{D}(y_1^n)=m}} \prod_{i=1}^n W(y_i|x_i)\Bigg\}\\
&=\frac{1}{|\mathcal{M}|}\sum_{m\in\mathcal{M}} \max_{x_1^n\in\mathcal{X}^n}\Bigg\{ \sum_{\substack{y_1^n\in\mathcal{Y}^n,\\\mathcal{D}(y_1^n)=m}} \prod_{i=1}^n \Bigg(\sum_{x_i'\in\mathcal{X}'} W'(y_i|x_i')V'(x_i'|x_i)\Bigg)\Bigg\}\\
&=\frac{1}{|\mathcal{M}|}\sum_{m\in\mathcal{M}} \max_{x_1^n\in\mathcal{X}^n}\Bigg\{ \sum_{\substack{y_1^n\in\mathcal{Y}^n,\\\mathcal{D}(y_1^n)=m}} \sum_{x_1'^n\in\mathcal{X}'^n} \prod_{i=1}^n \Bigg(W'(y_i|x_i')V'(x_i'|x_i)\Bigg)\Bigg\}\\
&=\frac{1}{|\mathcal{M}|}\sum_{m\in\mathcal{M}} \max_{x_1^n\in\mathcal{X}^n}\Bigg\{ \sum_{x_1'^n\in\mathcal{X}'^n} \Bigg(\prod_{i=1}^n V'(x_i'|x_i)\Bigg) \sum_{\substack{y_1^n\in\mathcal{Y}^n,\\\mathcal{D}(y_1^n)=m}}  \prod_{i=1}^n W'(y_i|x_i')\Bigg\}\\
&\leq \frac{1}{|\mathcal{M}|}\sum_{m\in\mathcal{M}} \max_{x_1'^n\in\mathcal{X}'^n}\Bigg\{ \sum_{\substack{y_1^n\in\mathcal{Y}^n,\\\mathcal{D}(y_1^n)=m}}  \prod_{i=1}^n W'(y_i|x_i')\Bigg\} = 1 - P_{e,\mathcal{D}}(W').
\end{align*}
Therefore $P_{e,\mathcal{D}}(W')\leq P_{e,\mathcal{D}}(W)$.

If $W$ and $W'$ are input-degraded from each other, then $P_{e,\mathcal{D}}(W')\leq P_{e,\mathcal{D}}(W)$ and $P_{e,\mathcal{D}}(W)\leq P_{e,\mathcal{D}}(W')$, hence $P_{e,\mathcal{D}}(W')= P_{e,\mathcal{D}}(W)$.
\end{proof}

\subsection{A characterization of input-degradedness}

Let $W\in\DMC_{\mathcal{X},\mathcal{Y}}$ and let $\mathcal{U}$ be a finite set. For every $p\in\Delta_{\mathcal{U}}$ and every $D\in\DMC_{\mathcal{Y},\mathcal{U}}$, define
$$P_c(p,W,D)=\sup_{E\in\DMC_{\mathcal{U},\mathcal{X}}}\sum_{\substack{u\in \mathcal{U},\\x\in\mathcal{X},\\y\in\mathcal{Y}}}p(u) E(x|u)W(y|x)D(u|y).$$

$P_c(p,W,D)$ can be interpreted as follows: let $U$ be a random variable in $\mathcal{U}$ distributed as $p$. Assume that $U$ was encoded using the random encoder $E\in\DMC_{\mathcal{U},\mathcal{X}}$ to get $X\in\mathcal{X}$. Send $X$ through the channel $W$ and let $Y\in\mathcal{Y}$ be the output. Apply the random decoder $D\in\DMC_{\mathcal{Y},\mathcal{U}}$ on $\mathcal{Y}$ to get an estimate $\hat{U}$ of $U$. We have:
$$\mathbb{P}[\{\hat{U}=U\}]=\sum_{\substack{u\in \mathcal{U},\\x\in\mathcal{X},\\y\in\mathcal{Y}}}p(u) E(x|u)W(y|x)D(u|y).$$
Therefore, $P_c(p,W,D)$ is the optimal probability of successfully estimating $U$ by the fixed decoder $D$ among all random encoders $E\in\DMC_{\mathcal{U},\mathcal{X}}$. Note that the optimal encoder can always be chosen to be deterministic.

\begin{mythe}
\label{theCharacInputDegrad}
A channel $W\in\DMC_{\mathcal{X},\mathcal{Y}}$ is input-degraded from another channel $W'\in\DMC_{\mathcal{X}',\mathcal{Y}}$ if and only if $P_c(p,W,D)\leq P_c(p,W',D)$ for every $p\in\Delta_{\mathcal{U}}$, every $D\in\DMC_{\mathcal{Y},\mathcal{U}}$ and every finite set $\mathcal{U}$.
\end{mythe}
\begin{proof}
Assume that $W$ is input-degraded from $W'$. There exists $V'\in\DMC_{\mathcal{X},\mathcal{X}'}$ such that $W=W'\circ V'$. For every finite set $\mathcal{U}$, every $p\in\Delta_{\mathcal{U}}$ and every $D\in\DMC_{\mathcal{Y},\mathcal{U}}$, we have:
\begin{align*}
P_c(p,W,D)&=\sup_{E\in\DMC_{\mathcal{U},\mathcal{X}}}\sum_{\substack{u\in \mathcal{U},\\x\in\mathcal{X},\\y\in\mathcal{Y}}}p(u) E(x|u)W(y|x)D(u|y)\\
&=\sup_{E\in\DMC_{\mathcal{U},\mathcal{X}}}\sum_{\substack{u\in \mathcal{U},\\x\in\mathcal{X},\\y\in\mathcal{Y}}}p(u)E(x|u)\Bigg(\sum_{x'\in\mathcal{X}'}W'(y|x') V'(x'|x)\Bigg)D(u|y)\\
&=\sup_{E\in\DMC_{\mathcal{U},\mathcal{X}}}\sum_{\substack{u\in \mathcal{U},\\x'\in\mathcal{X}',\\y\in\mathcal{Y}}}p(u) \Bigg(\sum_{x\in\mathcal{X}} V'(x'|x)E(x|u)\Bigg)W'(y|x')D(u|y)\\
&=\sup_{E\in\DMC_{\mathcal{U},\mathcal{X}}}\sum_{\substack{u\in \mathcal{U},\\x'\in\mathcal{X}',\\y\in\mathcal{Y}}}p(u) (V'\circ E)(x'|u)W'(y|x')D(u|y)\\
&\leq \sup_{E'\in\DMC_{\mathcal{U},\mathcal{X}'}}\sum_{\substack{u\in \mathcal{U},\\x'\in\mathcal{X}',\\y\in\mathcal{Y}}}p(u) E'(x'|u)W'(y|x')D(u|y) = P_c(p,W',D).
\end{align*}

Conversely, assume that $P_c(p,W,D)\leq P_c(p,W',D)$ for every $p\in\Delta_{\mathcal{U}}$, every $D\in\DMC_{\mathcal{Y},\mathcal{U}}$ and every finite set $\mathcal{U}$.

Let $x_0$ be any symbol that does belong to $\mathcal{X}$ and let $\mathcal{U}=\mathcal{X}\cup\{x_0\}$. For every $n\geq 1$, define $p_n\in \Delta_{\mathcal{U}}$ as follows:
$$p_n(u)=\begin{cases}\displaystyle\frac{1}{|\mathcal{X}|}\left(1-\frac{1}{n+1}\right)\quad&\text{if}\;u\in\mathcal{X},\\\displaystyle\frac{1}{n+1}\quad&\text{if} \;u=x_0. \end{cases}$$
$p_n$ was chosen in such a way that $\frac{p_n(x_0)}{p_n(x)}=\frac{|\mathcal{X}|}{n}$ for every $x\in\mathcal{X}$. This is going to be useful later. Define the channel $W_0\in\DMC_{\mathcal{U},\mathcal{Y}}$ as follows:
$$W_0(y|u)=\begin{cases}W(y|u)\quad&\text{if}\;u\in\mathcal{X},\\\displaystyle\frac{1}{|\mathcal{X}|}\sum_{x\in\mathcal{X}}W(y|x)\quad&\text{if}\;u=x_0. \end{cases}$$

Fix the encoder $E\in\DMC_{\mathcal{U},\mathcal{X}}$ as follows:
$$E(x|u)=\begin{cases}1\quad&\text{if}\;u=x,\\\displaystyle\frac{1}{|\mathcal{X}|}\quad&\text{if}\;u=x_0, \\0\quad&\text{otherwise}.\end{cases}$$
For every $D\in\DMC_{\mathcal{Y},\mathcal{U}}$, we have:
\begin{align*}
\sum_{\substack{u\in \mathcal{U},\\y\in\mathcal{Y}}}p_n(u) &W_0(y|u)D(u|y)\\
&=\Bigg(\sum_{\substack{x\in \mathcal{X},\\y\in\mathcal{Y}}}p_n(x) W_0(y|x)D(x|y)\Bigg) + \sum_{y\in\mathcal{Y}} p_n(x_0)W_0(y|x_0)D(x_0|y)\\
&=\Bigg(\sum_{\substack{x\in \mathcal{X},\\y\in\mathcal{Y}}}p_n(x) W(y|x)D(x|y)\Bigg) + \sum_{y\in\mathcal{Y}} p_n(x_0)\frac{1}{|\mathcal{X}|} \sum_{x\in\mathcal{X}} W(y|x)D(x_0|y)\\
&=\Bigg(\sum_{\substack{u\in\mathcal{X},\\ x\in \mathcal{X},\\y\in\mathcal{Y}}}p_n(u) E(x|u) W(y|x)D(u|y)\Bigg) + \sum_{\substack{x\in\mathcal{X},\\y\in\mathcal{Y}}} p_n(x_0)E(x|x_0)  W(y|x)D(x_0|y)\\
&=\sum_{\substack{u\in\mathcal{U},\\ x\in \mathcal{X},\\y\in\mathcal{Y}}}p_n(u) E(x|u) W(y|x)D(u|y)\leq P_c(p_n,W,D)\leq P_c(p_n,W',D)\\
&=\sup_{E'\in\DMC_{\mathcal{U},\mathcal{X}'}}\sum_{\substack{u\in \mathcal{U},\\x'\in\mathcal{X}',\\y\in\mathcal{Y}}}p_n(u) E'(x'|u)W'(y|x')D(u|y).
\end{align*}
Therefore,
\begin{align*}
\min_{E'\in\DMC_{\mathcal{U},\mathcal{X}'}}\sum_{\substack{u\in \mathcal{U},\\y\in\mathcal{Y}}} p_n(u)\left(W_0(y|x)-\sum_{x'\in\mathcal{X}'} E'(x'|u)W'(y|x')\right)D(u|y)\leq 0,
\end{align*}
hence
\begin{equation*}
\max_{D\in\DMC_{\mathcal{Y},\mathcal{U}}} \min_{E'\in\DMC_{\mathcal{U},\mathcal{X}'}}\sum_{\substack{u\in \mathcal{U},\\y\in\mathcal{Y}}} p_n(u)\left(W_0(y|u)-\sum_{x'\in\mathcal{X}'} E'(x'|u)W'(y|x')\right)D(u|y)\leq 0,
\end{equation*}
or equivalently
\begin{equation}
\label{eqMaxMinInputDegrad}
\max_{D\in\DMC_{\mathcal{Y},\mathcal{U}}} \min_{E'\in\DMC_{\mathcal{U},\mathcal{X}'}}\sum_{\substack{u\in \mathcal{U},\\y\in\mathcal{Y}}} p_n(u)\Big(W_0(y|u)-(W'\circ E')(y|u)\Big)D(u|y)\leq 0.
\end{equation}

Note that the sets $\DMC_{\mathcal{Y},\mathcal{U}}$ and $\DMC_{\mathcal{U},\mathcal{X}'}$ are compact and convex. On the other hand, since the function $\displaystyle \sum_{\substack{u\in \mathcal{U},\\y\in\mathcal{Y}}} p_n(u)\Big(W_0(y|u)-(W'\circ E')(y|u)\Big)D(u|y)$ is affine in both $D\in\DMC_{\mathcal{Y},\mathcal{U}}$ and $E'\in\DMC_{\mathcal{U},\mathcal{X}'}$, it is continuous, concave in $D$ and convex in $E'$. Therefore, we can apply the minimax theorem \cite{MiniMax} to exchange the max and the min in Equation \eqref{eqMaxMinInputDegrad}. We obtain:
\begin{equation*}
\min_{E'\in\DMC_{\mathcal{U},\mathcal{X}'}} \max_{D\in\DMC_{\mathcal{Y},\mathcal{U}}}\sum_{\substack{u\in \mathcal{U},\\y\in\mathcal{Y}}} p_n(u)\Big(W_0(y|u)-(W'\circ E')(y|u)\Big)D(u|y)\leq 0.
\end{equation*}
Therefore, there exists $E'_n\in\DMC_{\mathcal{U},\mathcal{X}'}$ such that
\begin{align*}
0&\geq \max_{D\in\DMC_{\mathcal{Y},\mathcal{U}}}\sum_{\substack{u\in \mathcal{U},\\y\in\mathcal{Y}}} p_n(u)\Big(W_0(y|u)-(W'\circ E'_n)(y|u)\Big)D(u|y)\\
&\stackrel{(a)}{=}\sum_{y\in\mathcal{Y}} \max_{u\in\mathcal{U}} p_n(u)\Big(W_0(y|u)-(W'\circ E'_n)(y|u)\Big)\\
&\geq \sum_{y\in\mathcal{Y}} \frac{1}{|\mathcal{U}|}\sum_{u\in\mathcal{U}} p_n(u) \Big(W_0(y|u)-(W'\circ E'_n)(y|u)\Big)\\
&=\frac{1}{|\mathcal{U}|}\sum_{u\in\mathcal{U}}p_n(u)\sum_{y\in\mathcal{Y}} \Big(W_0(y|u)-(W'\circ E'_n)(y|u)\Big) = 0,
\end{align*}
where (a) follows from the fact that $\displaystyle\sum_{\substack{u\in \mathcal{U},\\y\in\mathcal{Y}}} p_n(u)\Big(W_0(y|u)-(W'\circ E'_n)(y|u)\Big)D(u|y)$ is maximized when $D$ is chosen to be deterministic in such a way that for every $y\in\mathcal{Y}$, $D(u_y|y)=1$ for any $u_y\in\mathcal{U}$ satisfying $\displaystyle p_n(u_y)(W_0(y|u_y)-(W'\circ E'_n)(y|u_y))=\max_{u\in\mathcal{U}} \Big\{p_n(u)\big(W_0(y|u)-(W'\circ E'_n)(y|u)\big)\Big\}$. We conclude that
$$\sum_{y\in\mathcal{Y}} \max_{u\in\mathcal{U}} p_n(u)\Big(W_0(y|u)-(W'\circ E'_n)(y|u)\Big)=0.$$
Assume there exists $y\in\mathcal{Y}$ and $\tilde{u}\in\mathcal{U}$ such that $$p_n(u)\big( W_0(y|\tilde{u})-(W'\circ E'_n)(y|\tilde{u})\big)< \max_{u\in\mathcal{U}} p_n(u)\Big(W_0(y|u)-(W'\circ E'_n)(y|u)\Big).$$ In this case, we have
\begin{align*}
0&=\sum_{u\in\mathcal{U}}p_n(u)\sum_{y\in\mathcal{Y}} \Big(W_0(y|u)-(W'\circ E'_n)(y|u)\Big)\\
&=\sum_{y\in\mathcal{Y}}\sum_{u\in\mathcal{U}}p_n(u)\Big(W_0(y|u)-(W'\circ E'_n)(y|u)\Big)\\
&< \sum_{y\in\mathcal{Y}} |\mathcal{U}|\cdot \max_{u\in\mathcal{U}} p_n(u)\Big(W_0(y|u)-(W'\circ E'_n)(y|u)\Big)=0,
\end{align*}
which is a contradiction. Therefore, for every $y\in\mathcal{Y}$ and every $x\in\mathcal{X}$, we have
\begin{align*}
p_n(x)\big(W(y|x)-(W'\circ E'_n)(y|x)\big)&=p_n(x)\big(W_0(y|x)-(W'\circ E'_n)(y|x)\big)\\
&=\max_{u\in\mathcal{U}} p_n(u)\Big(W_0(y|u)-(W'\circ E'_n)(y|u)\Big)\\
&=p_n(x_0)\big( W_0(y|x_0)-(W'\circ E'_n)(y|x_0)\big),
\end{align*}
which implies that
\begin{align*}
\big|W(y|x)-(W'\circ E'_n)(y|x)\big|= \frac{p_n(x_0)}{p_n(x)}\big| W_0(y|x_0)-(W'\circ E'_n)(y|x_0)\big|\leq \frac{p_n(x_0)}{p_n(x)}=\frac{|\mathcal{X}|}{n}.
\end{align*}
Since the space $\DMC_{\mathcal{U},\mathcal{X}'}$ is compact, there exists a converging subsequence $(E'_{n_k})_{k\geq 0}$ of $(E'_n)_{n\geq 1}$. Let $E'$ be the limit of $(E'_{n_k})_{k\geq 0}$. For every $x\in\mathcal{X}$ and every $y\in\mathcal{Y}$, we have:
\begin{align*}
\big|W(y|x)-(W'\circ E')(y|x)\big|=\lim_{k\to\infty} \big|W(y|x)-(W'\circ E'_{n_k})(y|x)\big|\leq \lim_{k\to\infty} \frac{|\mathcal{X}|}{n_k}=0,
\end{align*}
which means that $W(y|x)=(W'\circ E')(y|x)$. Define $V'\in\DMC_{\mathcal{X},\mathcal{X}'}$ as $V'(x'|x)=E'(x'|x)$ for every $x\in\mathcal{X}$ and every $x'\in\mathcal{X}'$. For every $x\in\mathcal{X}$ and every $y\in\mathcal{Y}$, we have:
\begin{align*}
(W'\circ V')(y|x)=\sum_{x'\in\mathcal{X}'} W'(y|x')V'(x'|x)=\sum_{x'\in\mathcal{X}'} W'(y|x')E'(x'|x)=(W'\circ E')(y|x)=W(y|x).
\end{align*}
Therefore, $W=W'\circ V'$. We conclude that $W$ is input-degraded from $W'$.
\end{proof}

\subsection{A characterization in terms of randomized games}

A \emph{randomized game} is a 5-tuple $\mathcal{G}= (\mathcal{Z},\mathcal{X},\mathcal{Y}, l,W)$ such that $\mathcal{X},\mathcal{Y}$ and $\mathcal{Z}$ are finite sets, $l$ is a mapping from $\mathcal{Z}\times\mathcal{Y}$ to $\mathbb{R}$, and $W\in\DMC_{\mathcal{X},\mathcal{Y}}$. The mapping $l$ is called the \emph{payoff function} of the game $\mathcal{G}$, and the channel $W$ is called the \emph{randomizer} of $\mathcal{G}$. During the game, a player sees a symbol $z\in\mathcal{Z}$ and decides on a symbol $x\in\mathcal{X}$. A random symbol $y\in\mathcal{Y}$ is then randomly generated according to the conditional probability distribution $W(y|x)$ and the player gets the payoff $l(z,y)$.

A \emph{strategy} for the game $\mathcal{G}$ is a channel $S\in\DMC_{\mathcal{Z},\mathcal{X}}$. For every $z\in\mathcal{Z}$, the \emph{payoff gained by the strategy $S$ for $z$ in the game $\mathcal{G}$} is given by:
$$\$(z,S,\mathcal{G})=\sum_{\substack{x\in\mathcal{X},\\y\in\mathcal{Y}}}S(x|z)W(y|x)l(z,y).$$
The \emph{payoff vector gained by the strategy $S$ in the game $\mathcal{G}$} is given by:
$$\vec{\$}(S,\mathcal{G})=\big(\$(z,S,\mathcal{G})\big)_{z\in\mathcal{Z}}\in\mathbb{R}^{\mathcal{Z}}.$$
It is easy to see that for every $\alpha\in[0,1]$ and every $S_1,S_2\in\DMC_{\mathcal{Z},\mathcal{X}}$, we have $$\vec{\$}(\alpha S_1+(1-\alpha)S_2,\mathcal{G})=\alpha\vec{\$}(S_1,\mathcal{G})+(1-\alpha)\vec{\$}(S_2,\mathcal{G}).$$
The \emph{achievable payoff region for the game $\mathcal{G}$} is given by:
$$\$_{\ach}(\mathcal{G})=\Big\{\vec{\$}(S,\mathcal{G}):\; S\in{\DMC}_{\mathcal{Z},\mathcal{X}}\Big\}\subset \mathbb{R}^{\mathcal{Z}}.$$
Clearly, $\$_{\ach}(\mathcal{G})$ is a convex subset of $\mathbb{R}^{\mathcal{Z}}$. Moreover, since $\DMC_{\mathcal{Z},\mathcal{X}}$ is compact and since the mapping $S\rightarrow \vec{\$}(S,\mathcal{G})$ is a continuous mapping from $\DMC_{\mathcal{Z},\mathcal{X}}$ to $\mathbb{R}^{\mathcal{Z}}$, the region $\$_{\ach}(\mathcal{G})$ is a compact subset of $\mathbb{R}^{\mathcal{Z}}$.

The \emph{average payoff for the strategy $S\in\DMC_{\mathcal{Z},\mathcal{X}}$ for the game $\mathcal{G}$} is given by:
$$\hat{\$}(S,\mathcal{G})=\frac{1}{|\mathcal{Z}|}\sum_{z\in\mathcal{Z}} \$(z,S,\mathcal{G})= \sum_{\substack{z\in\mathcal{Z},\\x\in\mathcal{X},\\y\in\mathcal{Y}}}\frac{1}{|\mathcal{Z}|}S(x|z)W(y|x)l(z,y).$$
The \emph{optimal average payoff for the game $\mathcal{G}$} is given by
$$\$_{\opt}(\mathcal{G})= \sup_{S\in\DMC_{\mathcal{Z},\mathcal{X}}} \hat{\$}(S,\mathcal{G}).$$
Note that we can always find an optimal strategy that is deterministic.

The following theorem provides a characterization of input-degradedness that is similar to the famous Blackwell-Sherman-Stein theorem \cite{blackwell1951}, \cite{Sherman}, \cite{Stein}.
\begin{mythe}
\label{theGameOperatInputDegrad}
Let $\mathcal{X},\mathcal{X}'$ and $\mathcal{Y}$ be three finite sets. Let $W\in\DMC_{\mathcal{X},\mathcal{Y}}$ and $W'\in\DMC_{\mathcal{X}',\mathcal{Y}}$. The following conditions are equivalent:
\begin{itemize}
\item[(a)] $W$ is input-degraded from $W'$.
\item[(b)] For every finite set $\mathcal{Z}$ and every payoff function $l:\mathcal{Z}\times\mathcal{Y}\rightarrow\mathbb{R}$, we have
$$\$_{\ach}(\mathcal{Z},\mathcal{X},\mathcal{Y},l,W)\subset \$_{\ach}(\mathcal{Z},\mathcal{X}',\mathcal{Y},l,W').$$
\item[(c)] For every finite set $\mathcal{Z}$ and every payoff function $l:\mathcal{Z}\times\mathcal{Y}\rightarrow\mathbb{R}$, we have
$$\$_{\opt}(\mathcal{Z},\mathcal{X},\mathcal{Y},l,W)\leq \$_{\opt}(\mathcal{Z},\mathcal{X}',\mathcal{Y},l,W').$$
\end{itemize}
\end{mythe}
\begin{proof}
Assume that (a) is true. There exists $V'\in\DMC_{\mathcal{X}',\mathcal{X}}$ such that $W=W'\circ V'$. Fix a finite set $\mathcal{Z}$ and a payoff function $l:\mathcal{Z}\times\mathcal{Y}\rightarrow\mathbb{R}$. Define $\mathcal{G}=(\mathcal{Z},\mathcal{X},\mathcal{Y},l,W)$ and $\mathcal{G}'=(\mathcal{Z},\mathcal{X}',\mathcal{Y},l,W')$.

Fix $\vec{v}=(v_z)_{z\in\mathcal{Z}}\in \$_{\ach}(\mathcal{G})$. There exists $S\in\DMC_{\mathcal{Z},\mathcal{X}}$ such that $(v_z)_{z\in\mathcal{Z}}=\vec{v}=\big(\$(z,S,\mathcal{G})\big)_{z\in\mathcal{Z}}$. Let $S'=V'\circ S$. For every $z\in\mathcal{Z}$, we have:
\begin{align*}
\$(z,S',\mathcal{G}')&=\sum_{\substack{x'\in\mathcal{X}',\\y\in\mathcal{Y}}}S'(x'|z)W'(y|x')l(z,y)=\sum_{\substack{x'\in\mathcal{X}',\\y\in\mathcal{Y}}}\Bigg(\sum_{x\in\mathcal{X}} V'(x'|x)S(x|z)\Bigg) W'(y|x')l(z,y)\\
&=\sum_{\substack{x\in\mathcal{X},\\y\in\mathcal{Y}}} S(x|z) \Big(\sum_{x'\in\mathcal{X}'}W'(y|x')V'(x'|x)\Big)l(z,y)=\sum_{\substack{x\in\mathcal{X},\\y\in\mathcal{Y}}} S(x|z)W(y|x)l(z,y)=\$(z,S,\mathcal{G}).
\end{align*}
Therefore, $\vec{v}= \vec{\$}(S',\mathcal{G}')\in \$_{\ach}(\mathcal{G}')$. Since this is true for every $\vec{v}\in \$_{\ach}(\mathcal{G})$, we have $\$_{\ach}(\mathcal{G})\subset \$_{\ach}(\mathcal{G}')$. We conclude that (a) implies (b).

Now assume that (b) is true. Fix a finite set $\mathcal{Z}$ and a payoff function $l:\mathcal{Z}\times\mathcal{Y}\rightarrow\mathbb{R}$. Define $\mathcal{G}=(\mathcal{Z},\mathcal{X},\mathcal{Y},l,W)$ and $\mathcal{G}'=(\mathcal{Z},\mathcal{X}',\mathcal{Y},l,W')$. We have $\$_{\ach}(\mathcal{G})\subset \$_{\ach}(\mathcal{G}')$. Therefore,
\begin{align*}
\$_{\opt}(\mathcal{G})=\sup_{(v_z)_{z\in\mathcal{Z}} \in \$_{\ach}(\mathcal{G})} \frac{1}{|\mathcal{Z}|}\sum_{z\in\mathcal{Z}} v_z \stackrel{(\ast)}{\leq} \sup_{(v_z')_{z\in\mathcal{Z}} \in \$_{\ach}(\mathcal{G}')} \frac{1}{|\mathcal{Z}|} \sum_{z\in\mathcal{Z}} v_z' = \$_{\opt}(\mathcal{G}'),
\end{align*}
where $(\ast)$ follows from the fact that $\$_{\ach}(\mathcal{G})\subset \$_{\ach}(\mathcal{G}')$. This shows that (b) implies (c).

Now assume that (c) is true. Fix a finite set $\mathcal{U}$, $p\in\Delta_{\mathcal{U}}$ and $D\in\DMC_{\mathcal{Y},\mathcal{U}}$. Define the payoff function $l:\mathcal{U}\times\mathcal{Y}\rightarrow\mathbb{R}$ as $l(u,y)=|\mathcal{U}|p(u)D(u|y)$. Define the randomized games $\mathcal{G}=(\mathcal{U},\mathcal{X},\mathcal{Y},W,l)$ and $\mathcal{G}'=(\mathcal{U},\mathcal{X}',\mathcal{Y},W',l)$. We have:
\begin{align*}
P_c(p,W,D)&=\sup_{E\in\DMC_{\mathcal{U},\mathcal{X}}}\sum_{\substack{u\in \mathcal{U},\\x\in\mathcal{X},\\y\in\mathcal{Y}}}p(u) E(x|u)W(y|x)D(u|y)=\sup_{E\in\DMC_{\mathcal{U},\mathcal{X}}}\sum_{\substack{u\in \mathcal{U},\\x\in\mathcal{X},\\y\in\mathcal{Y}}}\frac{1}{|\mathcal{U}|} E(x|u)W(y|x)l(u,y)\\
&=\sup_{E\in\DMC_{\mathcal{U},\mathcal{X}}}\hat{\$}(E,\mathcal{G})=\$_{\opt}(\mathcal{G}).
\end{align*}

Similarly, we can show that $P_c(p,W',D)=\$_{\opt}(\mathcal{G}')$. Since we assumed that (c) is true, we have $\$_{\opt}(\mathcal{G})\leq \$_{\opt}(\mathcal{G}')$. Therefore, for every finite set $\mathcal{U}$, every $p\in\Delta_{\mathcal{U}}$ and every $D\in\DMC_{\mathcal{Y},\mathcal{U}}$, we have $P_c(p,W,D)\leq P_c(p,W',D)$. Theorem \ref{theCharacInputDegrad} now implies that $W$ is input-degraded from $W'$, hence (c) implies (a). We conclude that (a), (b) and (c) are equivalent.
\end{proof}

\section{Space of input-equivalent channels from $\mathcal{X}$ to $\mathcal{Y}$}

\subsection{The $\DMC_{\mathcal{X},\mathcal{Y}}^{(i)}$ space}

\label{subsecDMCXYi}
Let $\mathcal{X}$ and $\mathcal{Y}$ be two finite sets. Define the equivalence relation $R_{\mathcal{X},\mathcal{Y}}^{(i)}$ on $\DMC_{\mathcal{X},\mathcal{Y}}$ as follows:
$$WR_{\mathcal{X},\mathcal{Y}}^{(i)}W'\;\;\Leftrightarrow\;\;W\;\text{is input-equivalent to}\;W'.$$

\begin{mydef}
The space of input-equivalent channels with input alphabet $\mathcal{X}$ and output alphabet $\mathcal{Y}$ is the quotient of the space of channels from $\mathcal{X}$ to $\mathcal{Y}$ by the input-equivalence relation:
$$\textstyle\DMC_{\mathcal{X},\mathcal{Y}}^{(i)}=\DMC_{\mathcal{X},\mathcal{Y}}/R_{\mathcal{X},\mathcal{Y}}^{(i)}.$$
We define the topology $\mathcal{T}_{\mathcal{X},\mathcal{Y}}^{(i)}$ on $\DMC_{\mathcal{X},\mathcal{Y}}^{(i)}$ as the quotient topology $\mathcal{T}_{\mathcal{X},\mathcal{Y}}/R_{\mathcal{X},\mathcal{Y}}^{(i)}$.
\end{mydef}

Due to proposition \ref{propCharacInputEquiv}, we can define the \emph{input-equivalence characteristic of $\hat{W}\in\DMC_{\mathcal{X},\mathcal{Y}}^{(i)}$} as $\CE(\hat{W}):=\CE(W')$ for any $W'\in\hat{W}$. Define $\conv(\hat{W}):=\conv(\CE(\hat{W}))$. It is easy to see that $\conv(\hat{W})=\conv(\{W_x':\;x\in\mathcal{X}\})$ for any $W'\in\hat{W}$.

Let $A$ and $B$ be two sets. A \emph{coupling} of $A$ and $B$ is a subset $R$ of $A\times B$ such that
$$\{a\in A:\;\exists b\in B,\;(a,b)\in R\}=A,$$
and
$$\{b\in B:\;\exists a\in A,\;(a,b)\in R\}=B.$$
We denote the set of couplings of $A$ and $B$ as $\mathcal{R}(A,B)$.

We define the \emph{similarity distance} on $\DMC_{\mathcal{X},\mathcal{Y}}^{(i)}$ as follows:
\begin{align*}
d_{\mathcal{X},\mathcal{Y}}^{(i)}(\hat{W}_1,\hat{W}_2)&=\inf_{R\in \mathcal{R}(\conv(\hat{W}_1),\conv(\hat{W}_2))}\sup_{(P_1,P_2)\in R} \|P_1-P_2\|_{TV}\\
&=\frac{1}{2}\inf_{R\in \mathcal{R}(\conv(\hat{W}_1),\conv(\hat{W}_2))}\sup_{(P_1,P_2)\in R} \sum_{y\in\mathcal{Y}} |P_1(y)-P_2(y)|.
\end{align*}

\begin{myprop}
$(\DMC_{\mathcal{X},\mathcal{Y}}^{(i)},d_{\mathcal{X},\mathcal{Y}}^{(i)})$ is a metric space.
\end{myprop}
\begin{proof}
We will show that $d_{\mathcal{X},\mathcal{Y}}^{(i)}(\hat{W}_1, \hat{W}_2)=d_H\big(\conv(\hat{W}_1),\conv(\hat{W}_2)\big)$, where $d_H$ is the Hausdorff metric on $\mathcal{K}(\Delta_{\mathcal{Y}})$ corresponding to the total variation distance on $\Delta_{\mathcal{Y}}$. Define $K_1=\conv(\hat{W}_1)$ and $K_2=\conv(\hat{W}_2)$, and let $R\in\mathcal{R}(K_1,K_2)$. For every $(P_1,P_2)\in R$, we have:
$$\|P_1-P_2\|_{TV}\geq \inf_{P_2'\in K_2} \|P_1-P_2'\|_{TV}.$$
Therefore,
$$\sup_{(P_1,P_2)\in R} \|P_1-P_2\|_{TV}\geq \sup_{P_1'\in K_1} \inf_{P_2'\in K_2} \|P_1'-P_2'\|_{TV}.$$
Similarly,
$$\sup_{(P_1,P_2)\in R} \|P_1-P_2\|_{TV}\geq \sup_{P_2'\in K_2} \inf_{P_1'\in K_1} \|P_1'-P_2'\|_{TV}.$$
Hence,
\begin{align*}
\sup_{(P_1,P_2)\in R} \|P_1-P_2\|_{TV}&\geq \max\left\{ \sup_{P_1'\in K_1} \inf_{P_2'\in K_2} \|P_1'-P_2'\|_{TV},\sup_{P_2'\in K_2} \inf_{P_1'\in K_1} \|P_1'-P_2'\|_{TV}\right\}\\
&=d_H(K_1,K_2).
\end{align*}
We conclude that
$$d_{\mathcal{X},\mathcal{Y}}^{(i)}(\hat{W}_1,\hat{W}_2)=\inf_{R\in\mathcal{R}(K_1,K_2)}\sup_{(P_1,P_2)\in R}\|P_1-P_2\|_{TV}\geq d_H(K_1,K_2).$$
Let $P_1\in K_1$. Since $K_2$ is compact, there exists $\tilde{P}_2(P_1)\in K_2$ such that $$\|P_1-\tilde{P}_2(P_1)\|_{TV}=\inf_{P_2\in K_2}\|P_1-P_2\|_{TV}.$$ Similarly, for every $P_2\in K_2$, there exists $\tilde{P}_1(P_2)\in K_1$ such that $\displaystyle\|P_2-\tilde{P}_1(P_2)\|_{TV}=\inf_{P_1\in K_1}\|P_1-P_2\|_{TV}$. Define the coupling $R_0\in\mathcal{R}(K_1,K_2)$ as
$$R_0=\{(P_1,\tilde{P}_2(P_1)):\;P_1\in K_1\}\cup \{(\tilde{P}_1(P_2),P_2):\;P_2\in K_2\}.$$
We have:
\begin{align*}
d_{\mathcal{X},\mathcal{Y}}^{(i)}(\hat{W}_1,\hat{W}_2)&=\inf_{R\in\mathcal{R}(K_1,K_2)}\sup_{(P_1,P_2)\in R}\|P_1-P_2\|_{TV}\leq \sup_{(P_1,P_2)\in R_0}\|P_1-P_2\|_{TV}\\
&=\max\left\{\sup_{P_1\in K_1}\|P_1-\tilde{P}_2(P_1)\|, \sup_{P_2\in K_2}\|P_2-\tilde{P}_1(P_2)\|\right\}=d_H(K_1,K_2).
\end{align*}
We conclude that $d_{\mathcal{X},\mathcal{Y}}^{(i)}(\hat{W}_1, \hat{W}_2)=d_H(K_1,K_2)=d_H\big(\conv(\hat{W}_1),\conv(\hat{W}_2)\big)$, hence $d_{\mathcal{X},\mathcal{Y}}^{(i)}$ is a metric.
\end{proof}

\begin{myprop}
\label{propReldXYdXYi}
Let $W,W'\in\DMC_{\mathcal{X},\mathcal{Y}}$ and let $\hat{W}$ and $\hat{W}'$ be the $R_{\mathcal{X},\mathcal{Y}}^{(i)}$-equivalence classes of $W$ and $W'$ respectively. We have $d_{\mathcal{X},\mathcal{Y}}^{(i)}(\hat{W},\hat{W}')\leq d_{\mathcal{X},\mathcal{Y}}(W,W')$.
\end{myprop}
\begin{proof}
Define $R_0\subset \conv(\hat{W})\times \conv(\hat{W}')$ as follows:
$$R_0=\left\{\left(\sum_{x\in\mathcal{X}}\lambda_x W_x,\sum_{x\in\mathcal{X}}\lambda_x W_x'\right):\; \sum_{x\in\mathcal{X}}\lambda_x=1,\;\text{and}\;\lambda_x\geq 0,\;\forall x\in \mathcal{X}\right\}.$$
Clearly, $R_0$ is a coupling of $\conv(\hat{W})$ and $\conv(\hat{W}')$. For every $(P_1,P_2)\in R_0$, there exists $(\lambda_x)_{x\in\mathcal{X}}\in [0,1]^{\mathcal{X}}$ such that $\displaystyle\sum_{x\in\mathcal{X}}\lambda_x=1$, $P_1=\displaystyle \sum_{x\in\mathcal{X}}\lambda_x W_x$ and $P_2=\displaystyle \sum_{x\in\mathcal{X}}\lambda_x W_x'$. We have:
\begin{align*}
\|P_1-P_2\|_{TV}&=\left\|\left(\sum_{x\in\mathcal{X}}\lambda_x W_x\right)- \left(\sum_{x\in\mathcal{X}}\lambda_x W_x'\right)\right\|_{TV}=\left\|\sum_{x\in\mathcal{X}}\lambda_x (W_x - W_x')\right\|_{TV}\\
&\leq \sum_{x\in\mathcal{X}} \lambda_x \|W_x-W_x'\|_{TV}\leq \sup_{x\in\mathcal{X}}\|W_x-W_x'\|_{TV}=d_{\mathcal{X},\mathcal{Y}}(W,W').
\end{align*}
Therefore,
\begin{align*}
d_{\mathcal{X},\mathcal{Y}}^{(i)}(\hat{W},\hat{W}')&=\inf_{R\in\mathcal{R}(\conv(\hat{W}),\conv(\hat{W}'))}\sup_{(P_1,P_2)\in R}\|P_1-P_2\|_{TV}\leq \sup_{(P_1,P_2)\in R_0}\|P_1-P_2\|_{TV}\leq d_{\mathcal{X},\mathcal{Y}}(W,W').
\end{align*}
\end{proof}

\begin{mythe}
\label{theDMCXYi}
The topology induced by $d_{\mathcal{X},\mathcal{Y}}^{(i)}$ on $\DMC_{\mathcal{X},\mathcal{Y}}^{(i)}$ is the same as the quotient topology $\mathcal{T}_{\mathcal{X},\mathcal{Y}}^{(i)}$. Moreover, $(\DMC_{\mathcal{X},\mathcal{Y}}^{(i)},d_{\mathcal{X},\mathcal{Y}}^{(i)})$ is compact and path-connected.
\end{mythe}
\begin{proof}
Since $(\DMC_{\mathcal{X},\mathcal{Y}},d_{\mathcal{X},\mathcal{Y}})$ is compact and path-connected, the quotient space $(\DMC_{\mathcal{X},\mathcal{Y}}^{(i)},\mathcal{T}_{\mathcal{X},\mathcal{Y}}^{(i)})$ is compact and path-connected.

Define the mapping $\Proj:\DMC_{\mathcal{X},\mathcal{Y}}\rightarrow\DMC_{\mathcal{X},\mathcal{Y}}^{(i)}$ as $\Proj(W)=\hat{W}$, where $\hat{W}$ is the $R_{\mathcal{X},\mathcal{Y}}^{(i)}$-equivalence class of $W$. Proposition \ref{propReldXYdXYi} implies that $\Proj$ is a continuous mapping from $(\DMC_{\mathcal{X},\mathcal{Y}},d_{\mathcal{X},\mathcal{Y}})$ to $(\DMC_{\mathcal{X},\mathcal{Y}}^{(i)},d_{\mathcal{X},\mathcal{Y}}^{(i)})$. Since $\Proj(W)$ depends only on $\hat{W}$, Lemma \ref{lemQuotientFunction} implies that the transcendent mapping of $\Proj$ defined on the quotient space $(\DMC_{\mathcal{X},\mathcal{Y}}^{(i)},\mathcal{T}_{\mathcal{X},\mathcal{Y}}^{(i)})$ is continuous. But the transcendent mapping of $\Proj$ is nothing but the identity on $\DMC_{\mathcal{X},\mathcal{Y}}^{(i)}$. Therefore, the identity mapping $id$ on $\DMC_{\mathcal{X},\mathcal{Y}}^{(i)}$ is a continuous mapping from $(\DMC_{\mathcal{X},\mathcal{Y}}^{(i)},\mathcal{T}_{\mathcal{X},\mathcal{Y}}^{(i)})$ to $(\DMC_{\mathcal{X},\mathcal{Y}}^{(i)},d_{\mathcal{X},\mathcal{Y}}^{(i)})$.
For every subset $U$ of $\DMC_{\mathcal{X},\mathcal{Y}}^{(i)}$ we have:
\begin{itemize}
\item If $U$ is open in $(\DMC_{\mathcal{X},\mathcal{Y}}^{(i)},d_{\mathcal{X},\mathcal{Y}}^{(i)})$, then $U=id^{-1}(U)$ is open in $(\DMC_{\mathcal{X},\mathcal{Y}}^{(i)},\mathcal{T}_{\mathcal{X},\mathcal{Y}}^{(i)})$.
\item If $U$ is open in $(\DMC_{\mathcal{X},\mathcal{Y}}^{(i)},\mathcal{T}_{\mathcal{X},\mathcal{Y}}^{(i)})$, then its complement $U^c$ is closed in $(\DMC_{\mathcal{X},\mathcal{Y}}^{(i)},\mathcal{T}_{\mathcal{X},\mathcal{Y}}^{(i)})$ which is compact, hence $U^c$ is compact in $(\DMC_{\mathcal{X},\mathcal{Y}}^{(i)},\mathcal{T}_{\mathcal{X},\mathcal{Y}}^{(i)})$. This shows that $U^c=id(U^c)$ is a compact subset of $(\DMC_{\mathcal{X},\mathcal{Y}}^{(i)},d_{\mathcal{X},\mathcal{Y}}^{(i)})$. But $(\DMC_{\mathcal{X},\mathcal{Y}}^{(i)},d_{\mathcal{X},\mathcal{Y}}^{(i)})$ is a metric space, so $U^c$ is closed in $(\DMC_{\mathcal{X},\mathcal{Y}}^{(i)},d_{\mathcal{X},\mathcal{Y}}^{(i)})$. Therefore, $U$ is open $(\DMC_{\mathcal{X},\mathcal{Y}}^{(i)},d_{\mathcal{X},\mathcal{Y}}^{(i)})$.
\end{itemize}
We conclude that $(\DMC_{\mathcal{X},\mathcal{Y}}^{(i)},\mathcal{T}_{\mathcal{X},\mathcal{Y}}^{(i)})$ and $(\DMC_{\mathcal{X},\mathcal{Y}}^{(i)},d_{\mathcal{X},\mathcal{Y}}^{(i)})$ have the same open sets. Therefore, the topology induced by $d_{\mathcal{X},\mathcal{Y}}^{(i)}$ on $\DMC_{\mathcal{X},\mathcal{Y}}^{(i)}$ is the same as the quotient topology $\mathcal{T}_{\mathcal{X},\mathcal{Y}}^{(i)}$. Now since $(\DMC_{\mathcal{X},\mathcal{Y}}^{(i)},\mathcal{T}_{\mathcal{X},\mathcal{Y}}^{(i)})$ is compact and path-connected, $(\DMC_{\mathcal{X},\mathcal{Y}}^{(i)},d_{\mathcal{X},\mathcal{Y}}^{(i)})$ is compact and path-connected as well.
\end{proof}

\vspace*{3mm}

In the rest of this paper, we always associate $\DMC_{\mathcal{X},\mathcal{Y}}^{(i)}$ with the similarity metric $d_{\mathcal{X},\mathcal{Y}}^{(i)}$ and the quotient topology $\mathcal{T}_{\mathcal{X},\mathcal{Y}}^{(i)}$.

\subsection{Canonical embedding and canonical identification}

Let $\mathcal{X}_1,\mathcal{X}_2$ and $\mathcal{Y}$ be three finite sets such that $|\mathcal{X}_1|\leq |\mathcal{X}_2|$. We will show that there is a canonical embedding from $\DMC_{\mathcal{X}_1,\mathcal{Y}}^{(i)}$ to $\DMC_{\mathcal{X}_2,\mathcal{Y}}^{(i)}$. In other words, there exists an explicitly constructable compact subset $A$ of $\DMC_{\mathcal{X}_2,\mathcal{Y}}^{(i)}$ such that $A$ is homeomorphic to $\DMC_{\mathcal{X}_1,\mathcal{Y}}^{(i)}$. $A$ and the homeomorphism depend only on $\mathcal{X}_1,\mathcal{X}_2$ and $\mathcal{Y}$ (this is why we say that they are canonical). Moreover, we can show that $A$ depends only on $|\mathcal{X}_1|$, $\mathcal{X}_2$ and $\mathcal{Y}$.

\begin{mylem}
\label{lemEquivChannelSurj}
For every $W\in\DMC_{\mathcal{X}_1,\mathcal{Y}}$ and every surjection $f$ from $\mathcal{X}_2$ to $\mathcal{X}_1$, $W$ is input-equivalent to $W\circ D_f$.
\end{mylem}
\begin{proof}
Clearly $W\circ D_f$ is input-degraded from $W$. Now let $f'$ be any mapping from $\mathcal{X}_1$ to $\mathcal{X}_2$ such that $f(f'(x_1))=x_1$ for every $x_1\in\mathcal{X}_1$. We have $W=W\circ (D_{f}\circ D_{f'})=(W\circ D_f)\circ D_{f'}$, and so $W$ is also input-degraded from $W\circ D_f$.
\end{proof}

\begin{mycor}
\label{corEquivChannelSurj}
For every $W,W'\in\DMC_{\mathcal{X}_1,\mathcal{Y}}$ and every two surjections $f,g$ from $\mathcal{X}_2$ to $\mathcal{X}_1$, we have:
$$W R_{\mathcal{X}_1,\mathcal{Y}}^{(i)} W'\;\;\Leftrightarrow\;\; (W\circ D_f)R_{\mathcal{X}_2,\mathcal{Y}}^{(i)}(W'\circ D_g).$$
\end{mycor}
\begin{proof}
Since $W$ is input-equivalent to $W\circ D_f$ and $W'$ is input-equivalent to $W'\circ D_g$, then $W$ is input-equivalent to $W'$ if and only if $W\circ D_f$ is input-equivalent to $W'\circ D_g$.
\end{proof}

\vspace*{3mm}

For every $W\in\DMC_{\mathcal{X}_1,\mathcal{Y}}$, we denote the $R_{\mathcal{X}_1,\mathcal{Y}}^{(i)}$-equivalence class of $W$ as $\hat{W}$, and for every $W\in\DMC_{\mathcal{X}_2,\mathcal{Y}}$, we denote the $R_{\mathcal{X}_2,\mathcal{Y}}^{(i)}$-equivalence class of $W$ as $\tilde{W}$.

\begin{myprop}
\label{propEmbedInEquiv}
Let $\mathcal{X}_1,\mathcal{X}_2$ and $\mathcal{Y}$ be three finite sets such that $|\mathcal{X}_1|\leq |\mathcal{X}_2|$. Let $f:\mathcal{X}_2\rightarrow\mathcal{X}_1$ be any fixed surjection from $\mathcal{X}_2$ to $\mathcal{X}_1$. Define the mapping $F:\DMC_{\mathcal{X}_1,\mathcal{Y}}^{(i)}\rightarrow \DMC_{\mathcal{X}_2,\mathcal{Y}}^{(i)}$ as
$F(\hat{W})=\widetilde{W'\circ D_f}=\Proj_2(W'\circ D_f)$, where $W'\in \hat{W}$ and $\Proj_2$ is the projection onto the $R_{\mathcal{X},\mathcal{Y}_2}^{(i)}$-equivalence classes. We have:
\begin{itemize}
\item $F$ is well defined, i.e., $F(\hat{W})$ does not depend on $W'\in\hat{W}$.
\item $F$ is a homeomorphism from $\DMC_{\mathcal{X}_1,\mathcal{Y}}^{(i)}$ to $F\big(\DMC_{\mathcal{X}_1,\mathcal{Y}}^{(i)}\big)\subset \DMC_{\mathcal{X}_2,\mathcal{Y}}^{(i)}$.
\item $F$ does not depend on the surjection $f$. It depends only on $\mathcal{X}_1$, $\mathcal{X}_2$ and $\mathcal{Y}$, hence it is canonical.
\item $F\big(\DMC_{\mathcal{X}_1,\mathcal{Y}}^{(i)}\big)$ depends only on $|\mathcal{X}_1|$, $\mathcal{X}_2$ and $\mathcal{Y}$.
\item For every $W'\in\hat{W}$ and every $W''\in F(\hat{W})$, $W'$ is input-equivalent to $W''$.
\end{itemize}
\end{myprop}
\begin{proof}
Corollary \ref{corEquivChannelSurj} implies that $\Proj_2(W\circ D_f)=\Proj_2(W'\circ D_f)$ if and only if $W R_{\mathcal{X}_1,\mathcal{Y}}^{(i)}W'$. Therefore, $\Proj_2(W'\circ D_f)$ does not depend on $W'\in\hat{W}$, hence $F$ is well defined. Corollary \ref{corEquivChannelSurj} also shows that $\Proj_2(W'\circ D_f)$ does not depend on the particular choice of the surjection $f$, hence it is canonical (i.e., it depends only on $\mathcal{X}_1,\mathcal{X}_2$ and $\mathcal{Y}$).

On the other hand, the mapping $W\rightarrow W\circ D_f$ is a continuous mapping from $\DMC_{\mathcal{X}_1,\mathcal{Y}}$ to $\DMC_{\mathcal{X}_2,\mathcal{Y}}$, and $\Proj_2$ is continuous. Therefore, the mapping $W\rightarrow \Proj_2(W\circ D_f)$ is a continuous mapping from $\DMC_{\mathcal{X}_1,\mathcal{Y}}$ to $\DMC_{\mathcal{X}_2,\mathcal{Y}}^{(i)}$. Now since $\Proj_2(W \circ D_f)$ depends only on the $R_{\mathcal{X}_1,\mathcal{Y}}^{(i)}$-equivalence class $\hat{W}$ of $W$, Lemma \ref{lemQuotientFunction} implies that the transcendent mapping of $W\rightarrow \Proj_2(W\circ D_f)$ that is defined on $\DMC_{\mathcal{X}_1,\mathcal{Y}}^{(i)}$ is continuous. Therefore, $F$ is a continuous mapping from $(\DMC_{\mathcal{X}_1,\mathcal{Y}}^{(i)},\mathcal{T}_{\mathcal{X}_1,\mathcal{Y}}^{(i)})$ to $(\DMC_{\mathcal{X}_2,\mathcal{Y}}^{(i)},\mathcal{T}_{\mathcal{X}_2,\mathcal{Y}}^{(i)})$. Moreover, we can see from Corollary \ref{corEquivChannelSurj} that $F$ is an injection.

For every closed subset $B$ of $\DMC_{\mathcal{X}_1,\mathcal{Y}}^{(i)}$, $B$ is compact since $\DMC_{\mathcal{X}_1,\mathcal{Y}}^{(i)}$ is compact, hence $F(B)$ is compact because $F$ is continuous. This implies that $F(B)$ is closed in $\DMC_{\mathcal{X}_2,\mathcal{Y}}^{(i)}$ since $\DMC_{\mathcal{X}_2,\mathcal{Y}}^{(i)}$ is Hausdorff (as it is metrizable). Therefore, $F$ is a closed mapping.

Now since $F$ is an injection that is both continuous and closed, $F$ is a homeomorphism between $\DMC_{\mathcal{X}_1,\mathcal{Y}}^{(i)}$ and $F\big(\DMC_{\mathcal{X}_1,\mathcal{Y}}^{(i)}\big)\subset \DMC_{\mathcal{X}_2,\mathcal{Y}}^{(i)}$.

We would like now to show that $F\big(\DMC_{\mathcal{X}_1,\mathcal{Y}}^{(i)}\big)$ depends only on $|\mathcal{X}_1|$, $\mathcal{X}_2$ and $\mathcal{Y}$. Let $\mathcal{X}_1'$ be a finite set such that $|\mathcal{X}_1|=|\mathcal{X}_1'|$. For every $W\in \DMC_{\mathcal{X}_1',\mathcal{Y}}$, let $\overline{W}\in\DMC_{\mathcal{X}_1',\mathcal{Y}}^{(i)}$ be the $R_{\mathcal{X}_1',\mathcal{Y}}^{(i)}$-equivalence class of $W$.

Let $g:\mathcal{X}_1\rightarrow \mathcal{X}_1'$ be a fixed bijection from $\mathcal{X}_1$ to $\mathcal{X}_1'$ and let $f'=g\circ f$. Define $F': \DMC_{\mathcal{X}_1',\mathcal{Y}}^{(i)}\rightarrow \DMC_{\mathcal{X}_2,\mathcal{Y}}^{(i)}$ as $F'(\overline{W})=\widetilde{W'\circ D_{f'}}=\Proj_2(W'\circ D_{f'}),$ where $W'\in \overline{W}$. As above, $F'$ is well defined, and it is a homeomorphism from $\DMC_{\mathcal{X}_1',\mathcal{Y}}^{(i)}$ to $F'\big(\DMC_{\mathcal{X}_1',\mathcal{Y}}^{(i)}\big)$. We want to show that $F'\big(\DMC_{\mathcal{X}_1',\mathcal{Y}}^{(i)}\big)=F\big(\DMC_{\mathcal{X}_1,\mathcal{Y}}^{(i)}\big)$. For every $\overline{W}\in \DMC_{\mathcal{X}_1',\mathcal{Y}}^{(i)}$, let $W'\in\overline{W}$. We have $$\textstyle F'(\overline{W})= \Proj_2(W'\circ D_{f'})=\Proj_2((W'\circ D_g)\circ D_f)=F\left(\widehat{W'\circ D_g}\right)\in F\big(\DMC_{\mathcal{X}_1,\mathcal{Y}}^{(i)}\big).$$ Since this is true for every $\overline{W}\in \DMC_{\mathcal{X}_1',\mathcal{Y}}^{(i)}$, we deduce that $F'\big(\DMC_{\mathcal{X}_1',\mathcal{Y}}^{(i)}\big)\subset F\big(\DMC_{\mathcal{X}_1,\mathcal{Y}}^{(i)}\big)$. By exchanging the roles of $\mathcal{X}_1$ and $\mathcal{X}_1'$ and using the fact that $f=g^{-1}\circ f'$, we get $F\big(\DMC_{\mathcal{X}_1,\mathcal{Y}}^{(i)}\big)\subset F'\big(\DMC_{\mathcal{X}_1',\mathcal{Y}}^{(i)}\big)$. We conclude that $F\big(\DMC_{\mathcal{X}_1,\mathcal{Y}}^{(i)}\big)=F'\big(\DMC_{\mathcal{X}_1',\mathcal{Y}}^{(i)}\big)$, which means that $F\big(\DMC_{\mathcal{X}_1,\mathcal{Y}}^{(i)}\big)$ depends only on $|\mathcal{X}_1|$, $\mathcal{X}_2$ and $\mathcal{Y}$.

For every $W'\in\hat{W}$ and every $W''\in F(\hat{W})=\widetilde{W'\circ D_f}$, $W''$ is input-equivalent to $W'\circ D_f$ and $W'\circ D_f$ is input-equivalent to $W'$ (by Lemma \ref{lemEquivChannelSurj}), hence $W''$ is input-equivalent to $W'$.
\end{proof}

\begin{mycor}
\label{corIdentInEquiv}
If $|\mathcal{X}_1|=|\mathcal{X}_2|$, there exists a canonical homeomorphism from $\DMC_{\mathcal{X}_1,\mathcal{Y}}^{(i)}$ to $\DMC_{\mathcal{X}_2,\mathcal{Y}}^{(i)}$ depending only on $\mathcal{X}_1,\mathcal{X}_2$ and $\mathcal{Y}$.
\end{mycor}
\begin{proof}
Let $f$ be a bijection from $\mathcal{X}_2$ to $\mathcal{X}_1$. Define the mapping $F:\DMC_{\mathcal{X}_1,\mathcal{Y}}^{(i)}\rightarrow \DMC_{\mathcal{X}_2,\mathcal{Y}}^{(i)}$ as
$F(\hat{W})=\widetilde{W'\circ D_f}=\Proj_2(W'\circ D_f),$ where $W'\in \hat{W}$ and $\Proj_2:\DMC_{\mathcal{X}_2,\mathcal{Y}}\rightarrow\DMC_{\mathcal{X}_2,\mathcal{Y}}^{(i)}$ is the projection onto the $R_{\mathcal{X}_2,\mathcal{Y}}^{(i)}$-equivalence classes.

Also, define the mapping $F':\DMC_{\mathcal{X}_2,\mathcal{Y}}^{(i)}\rightarrow \DMC_{\mathcal{X}_1,\mathcal{Y}}^{(i)}$ as
$F(\tilde{V})=\widehat{V'\circ D_{f^{-1}}}=\Proj_1(V'\circ D_{f^{-1}}),$ where $V'\in \tilde{V}$ and $\Proj_1:\DMC_{\mathcal{X},\mathcal{Y}_1}\rightarrow\DMC_{\mathcal{X}_1,\mathcal{Y}}^{(i)}$ is the projection onto the $R_{\mathcal{X}_1,\mathcal{Y}}^{(i)}$-equivalence classes.

Proposition \ref{propEmbedInEquiv} shows that $F$ and $F'$ are well defined.

For every $W\in\DMC_{\mathcal{X}_1,\mathcal{Y}}$, we have:
\begin{align*}
F'(F(\hat{W}))&\stackrel{(a)}{=}F'(\widetilde{W\circ D_f})\stackrel{(b)}{=}\widehat{(W\circ D_f)\circ D_{f^{-1}}}=\hat{W},
\end{align*}
where (a) follows from the fact that $W\in \hat{W}$ and (b) follows from the fact that $W\circ D_f\in \widetilde{W\circ D_f}$.

We can similarly show that $F(F'(\tilde{V}))=\tilde{V}$ for every $\tilde{V}\in\DMC_{\mathcal{X}_2,\mathcal{Y}}^{(i)}$. Therefore, both $F$ and $F'$ are bijections. Proposition \ref{propEmbedInEquiv} now implies that $F$ is a homeomorphism from $\DMC_{\mathcal{X}_1,\mathcal{Y}}^{(i)}$ to $F\big(\DMC_{\mathcal{X}_1,\mathcal{Y}}^{(i)}\big)=\DMC_{\mathcal{X}_2,\mathcal{Y}}^{(i)}$. Moreover, $F$ depends only on $\mathcal{X},\mathcal{Y}_1$ and $\mathcal{Y}_2$.
\end{proof}

\vspace*{3mm}

Corollary \ref{corIdentInEquiv} allows us to identify $\DMC_{\mathcal{X},\mathcal{Y}}^{(i)}$ with $\DMC_{[n],\mathcal{Y}}^{(i)}$ through the canonical homeomorphism, where $n=|\mathcal{X}|$ and $[n]=\{1,\ldots,n\}$. Moreover, for every $1\leq n\leq m$, Proposition \ref{propEmbedInEquiv} allows us to identify $\DMC_{[n],\mathcal{Y}}^{(i)}$ with the canonical subspace of $\DMC_{[m],\mathcal{Y}}^{(i)}$ that is homeomorphic to $\DMC_{[n],\mathcal{Y}}^{(i)}$. In the rest of this paper, we consider that $\DMC_{[n],\mathcal{Y}}^{(i)}$ is a compact subspace of $\DMC_{[m],\mathcal{Y}}^{(i)}$.

\vspace*{3mm}
Intuitively, $\DMC_{[n],\mathcal{Y}}^{(i)}$ has a ``lower dimension" compared to $\DMC_{[m],\mathcal{Y}}^{(i)}$. So one expects that the interior of $\DMC_{[n],\mathcal{Y}}^{(i)}$ in $(\DMC_{[m],\mathcal{Y}}^{(i)},\mathcal{T}_{[m],\mathcal{Y}}^{(i)})$ is empty if $m>n$. The following proposition shows that this intuition is accurate when $|\mathcal{Y}|\geq 3$.

\begin{myprop}
\label{propInteriorEmptyDMCXni}
We have:
\begin{itemize}
\item If $|\mathcal{Y}|=1$, then $\DMC_{[n],\mathcal{Y}}^{(i)}=\DMC_{[1],\mathcal{Y}}^{(i)}$ for every $n\geq 1$.
\item If $|\mathcal{Y}|=2$, then $\DMC_{[n],\mathcal{Y}}^{(i)}=\DMC_{[2],\mathcal{Y}}^{(i)}$ for every $n\geq 2$.
\item If $|\mathcal{Y}|\geq 3$, then for every $1\leq n<m$, the interior of $\DMC_{[n],\mathcal{Y}}^{(i)}$ in $(\DMC_{[m],\mathcal{Y}}^{(i)},\mathcal{T}_{[m],\mathcal{Y}}^{(i)})$ is empty.
\end{itemize}
\end{myprop}
\begin{proof}
See Appendix \ref{appInteriorEmptyDMCXni}.
\end{proof}

\section{Space of input-equivalent channels}

\label{secInEquivSpace}

The previous section showed that if we are interested in input-equivalent channels, it is sufficient to study the spaces $\DMC_{[n],\mathcal{Y}}$ and $\DMC_{[n],\mathcal{Y}}^{(i)}$ for every $n\geq 1$, where $[n]=\{1,\ldots,n\}$. Define the space $$\textstyle\DMC_{\ast,\mathcal{Y}}={\displaystyle\coprod_{n\geq 1}} \DMC_{[n],\mathcal{Y}},$$
where $\coprod$ is the disjoint union symbol. The subscript $\ast$ indicates that the input alphabets of the considered channels are arbitrary but finite. We define the equivalence relation $R_{\ast,\mathcal{Y}}^{(i)}$ on $\DMC_{\ast,\mathcal{Y}}$ as follows:
$$WR_{\ast,\mathcal{Y}}^{(i)}W'\;\;\Leftrightarrow\;\;W\;\text{is input-equivalent to}\;W'.$$

\begin{mydef}
The space of input-equivalent channels with output alphabet $\mathcal{Y}$ is the quotient of the space of channels with output alphabet $\mathcal{Y}$ by the input-equivalence relation:
$$\textstyle\DMC_{\ast,\mathcal{Y}}^{(i)}=\DMC_{\ast,\mathcal{Y}}/R_{\ast,\mathcal{Y}}^{(i)}.$$
\end{mydef}

Clearly, $\DMC_{[n],\mathcal{Y}}/R_{\ast,\mathcal{Y}}^{(i)}$ can be canonically identified with $\DMC_{[n],\mathcal{Y}}/R_{[n],\mathcal{Y}}^{(i)}=\DMC_{[n],\mathcal{Y}}^{(i)}$. Therefore, we can write
$$\textstyle\DMC_{\ast,\mathcal{Y}}^{(i)}={\displaystyle\bigcup_{n\geq 1}}\DMC_{[n],\mathcal{Y}}^{(i)}.$$

We define the \emph{input-rank} of $\hat{W}\in\DMC_{\ast,\mathcal{Y}}^{(i)}$ as the size of its characteristic: $\irank(\hat{W})=|\CE(\hat{W})|$. Due to Proposition \ref{propCharacInputEquiv}, we have
$$\textstyle\DMC_{[n],\mathcal{Y}}^{(i)}=\{\hat{W}\in\DMC_{\ast,\mathcal{Y}}^{(i)}:\;\irank(\hat{W})\leq n\}.$$

A subset $A$ of $\DMC_{\ast,\mathcal{Y}}^{(i)}$ is said to be \emph{rank-bounded} if there exists $n\geq 1$ such that $A\subset \DMC_{[n],\mathcal{Y}}^{(i)}$.

\subsection{Natural topologies on $\DMC_{\ast,\mathcal{Y}}^{(i)}$}

\label{subsecNaturalTop}

Since $\DMC_{\ast,\mathcal{Y}}^{(i)}$ is the quotient of $\DMC_{\ast,\mathcal{Y}}$ and since $\DMC_{\ast,\mathcal{Y}}$ was not given any topology, there is no ``standard topology" on $\DMC_{\ast,\mathcal{Y}}^{(i)}$. However, there are many properties that one may require from any ``reasonable" topology on $\DMC_{\ast,\mathcal{Y}}^{(i)}$. In this paper, we focus on one particular requirement that we consider the most basic property required from any ``acceptable" topology on $\DMC_{\ast,\mathcal{Y}}^{(i)}$:

\begin{mydef}
A topology $\mathcal{T}$ on $\DMC_{\ast,\mathcal{Y}}^{(i)}$ is said to be \emph{natural} if it induces the quotient topology $\mathcal{T}_{[n],\mathcal{Y}}^{(i)}$ on $\DMC_{[n],\mathcal{Y}}^{(i)}$ for every $n\geq 1$.
\end{mydef}

The reason why we consider such topology as natural is because the quotient topology $\mathcal{T}_{[n],\mathcal{Y}}^{(i)}$ is the ``standard" and ``most natural" topology on $\DMC_{[n],\mathcal{Y}}^{(i)}$. Therefore, we do not want to induce any non-standard topology on $\DMC_{[n],\mathcal{Y}}^{(i)}$ by relativization.

\begin{myprop}
Every natural topology is $\sigma$-compact, separable and path-connected.
\end{myprop}
\begin{proof}
Since $\DMC_{\ast,\mathcal{Y}}^{(i)}$ is the countable union of compact and separable subspaces (namely $\{\DMC_{[n],\mathcal{Y}}^{(i)}\}_{n\geq 1}$), $\DMC_{\ast,\mathcal{Y}}^{(i)}$ is $\sigma$-compact and separable as well.

On the other hand, since $\displaystyle\bigcap_{n\geq 1}\textstyle\DMC_{[n],\mathcal{Y}}^{(i)}=\DMC_{[1],\mathcal{Y}}^{(i)}\neq\o$ and since $\DMC_{[n],\mathcal{Y}}^{(i)}$ is path-connected for every $n\geq1$, the union $\textstyle\DMC_{\ast,\mathcal{Y}}^{(i)}={\displaystyle\bigcup_{n\geq 0}}\DMC_{[n],\mathcal{Y}}^{(i)}$ is path-connected.
\end{proof}

\vspace*{3mm}

Proposition \ref{propInteriorEmptyDMCXni} implies that if $|\mathcal{Y}|=1$, then $\DMC_{\ast,\mathcal{Y}}^{(i)}=\DMC_{[1],\mathcal{Y}}^{(i)}$, and so the only natural topology on $\DMC_{\ast,\mathcal{Y}}^{(i)}$ is $\mathcal{T}_{[1],\mathcal{Y}}^{(i)}$. Similarly, if $|\mathcal{Y}|=2$, then $\DMC_{\ast,\mathcal{Y}}^{(i)}=\DMC_{[2],\mathcal{Y}}^{(i)}$, and the only natural topology on $\DMC_{\ast,\mathcal{Y}}^{(i)}$ is $\mathcal{T}_{[2],\mathcal{Y}}^{(i)}$. In the rest of this section, we investigate the properties of natural topologies when $|\mathcal{Y}|\geq 3$.

\begin{myprop}
\label{propNaturallyOpenUnboundedIn}
If $|\mathcal{Y}|\geq 3$ and $\mathcal{T}$ is a natural topology, every open set is rank-unbounded.
\end{myprop}
\begin{proof}
Assume to the contrary that there exists a non-empty open set $U\in \mathcal{T}$ such that $U\subset \DMC_{[n],\mathcal{Y}}^{(i)}$ for some $n\geq 1$. $U\cap \DMC_{[n+1],\mathcal{Y}}^{(i)}$ is open in $\DMC_{[n+1],\mathcal{Y}}^{(i)}$ because $\mathcal{T}$ is natural. On the other hand, $U\cap \DMC_{[n+1],\mathcal{Y}}^{(i)}\subset U\subset \DMC_{[n],\mathcal{Y}}^{(i)}$. Proposition \ref{propInteriorEmptyDMCXni} now implies that $U\cap \DMC_{[n+1],\mathcal{Y}}^{(i)}=\o$. Therefore,
$$\textstyle U=U\cap \DMC_{[n],\mathcal{Y}}^{(i)}\subset U\cap \DMC_{[n+1],\mathcal{Y}}^{(i)}=\o,$$ which is a contradiction.
\end{proof}

\begin{mycor}
\label{corInteriorEmptyDMCXni}
If $|\mathcal{Y}|\geq 3$ and $\mathcal{T}$ is a natural topology, then for every $n\geq 1$, the interior of $\DMC_{[n],\mathcal{Y}}^{(i)}$ in $(\DMC_{\ast,\mathcal{Y}}^{(i)},\mathcal{T})$ is empty.
\end{mycor}

\begin{myprop}
\label{propNaturaIsNotBaireIn}
If $|\mathcal{Y}|\geq 3$ and $\mathcal{T}$ is a Hausdorff natural topology, then $(\DMC_{\ast,\mathcal{Y}}^{(i)},\mathcal{T})$ is not a Baire space.
\end{myprop}
\begin{proof}
Fix $n\geq 1$. Since $\mathcal{T}$ is natural, $\DMC_{[n],\mathcal{Y}}^{(i)}$ is a compact subset of $(\DMC_{\ast,\mathcal{Y}}^{(i)},\mathcal{T})$. But $\mathcal{T}$ is Hausdorff, so $\DMC_{[n],\mathcal{Y}}^{(i)}$ is a closed subset of $(\DMC_{\ast,\mathcal{Y}}^{(i)},\mathcal{T})$. Therefore, $\DMC_{\ast,\mathcal{Y}}^{(i)}\setminus\DMC_{[n],\mathcal{Y}}^{(i)}$ is open.

On the other hand, Corollary \ref{corInteriorEmptyDMCXni} shows that the interior of $\DMC_{[n],\mathcal{Y}}^{(i)}$ in $(\DMC_{\ast,\mathcal{Y}}^{(i)},\mathcal{T})$ is empty. Therefore, $\DMC_{\ast,\mathcal{Y}}^{(i)}\setminus\DMC_{[n],\mathcal{Y}}^{(i)}$ is dense in $(\DMC_{\ast,\mathcal{Y}}^{(i)},\mathcal{T})$.

Now since
$$\bigcap_{n\geq 1} {\textstyle\left(\DMC_{\ast,\mathcal{Y}}^{(i)}\setminus\DMC_{[n],\mathcal{Y}}^{(i)}\right)}={\DMC}_{\ast,\mathcal{Y}}^{(i)}\setminus {\textstyle\left({\displaystyle\bigcup_{n\geq 1}}\DMC_{[n],\mathcal{Y}}^{(i)}\right)}=\o,$$
and since $\DMC_{\ast,\mathcal{Y}}^{(i)}\setminus\DMC_{[n],\mathcal{Y}}^{(i)}$ is open and dense in $(\DMC_{\ast,\mathcal{Y}}^{(i)},\mathcal{T})$ for every $n\geq 1$, we conclude that $(\DMC_{\ast,\mathcal{Y}}^{(i)},\mathcal{T})$ is not a Baire space.
\end{proof}

\begin{mycor}
\label{corNaturalNotCompleteIn}
If $|\mathcal{Y}|\geq 3$, no natural topology on $\DMC_{\ast,\mathcal{Y}}^{(i)}$ can be completely metrizable.
\end{mycor}
\begin{proof}
The corollary follows from Proposition \ref{propNaturaIsNotBaireIn} and the fact that every completely metrizable topology is both Hausdorff and Baire.
\end{proof}

\begin{myprop}
\label{propNaturalNotLocallyCompactIn}
If $|\mathcal{Y}|\geq 3$ and $\mathcal{T}$ is a Hausdorff natural topology, then $(\DMC_{\ast,\mathcal{Y}}^{(i)},\mathcal{T})$ is not locally compact anywhere, i.e., for every $\hat{W}\in \DMC_{\ast,\mathcal{Y}}^{(i)}$, there is no compact neighborhood of $\hat{W}$ in $(\DMC_{\ast,\mathcal{Y}}^{(i)},\mathcal{T})$.
\end{myprop}
\begin{proof}
Assume to the contrary that there exists a compact neighborhood $K$ of $\hat{W}$. There exists an open set $U$ such that $\hat{W}\in U\subset K$.

Since $K$ is compact and Hausdorff, it is a Baire space. Moreover, since $U$ is an open subset of $K$, $U$ is also a Baire space.

Fix $n\geq 1$. Since the interior of $\DMC_{[n],\mathcal{Y}}^{(i)}$ in $(\DMC_{\ast,\mathcal{Y}}^{(i)},\mathcal{T})$ is empty, the interior of $U\cap\DMC_{[n],\mathcal{Y}}^{(i)}$ in $U$ is also empty. Therefore, $U\setminus\DMC_{[n],\mathcal{Y}}^{(i)}$ is dense in $U$. On the other hand, since $\mathcal{T}$ is natural, $\DMC_{[n],\mathcal{Y}}^{(i)}$ is compact which implies that it is closed because $\mathcal{T}$ is Hausdorff. Therefore, $U\setminus\DMC_{[n],\mathcal{Y}}^{(i)}$ is open in $U$. Now since
$$\bigcap_{n\geq 1} {\textstyle\left(U\setminus\DMC_{[n],\mathcal{Y}}^{(i)}\right)}=U\setminus {\textstyle\left({\displaystyle\bigcup_{n\geq 1}}\DMC_{[n],\mathcal{Y}}^{(i)}\right)}=\o,$$
and since $U\setminus\DMC_{[n],\mathcal{Y}}^{(i)}$ is open and dense in $U$ for every $n\geq 1$, $U$ is not Baire, which is a contradiction. Therefore, there is no compact neighborhood of $\hat{W}$ in $(\DMC_{\ast,\mathcal{Y}}^{(i)},\mathcal{T})$.
\end{proof}

\section{Strong topology on $\DMC_{\ast,\mathcal{Y}}^{(i)}$}

\label{subsecInStrongTop}

The first natural topology that we study is the \emph{strong topology} $\mathcal{T}_{s,\ast,\mathcal{Y}}^{(i)}$ on $\DMC_{\ast,\mathcal{Y}}^{(i)}$, which is the finest natural topology.

Since the spaces $\{\DMC_{[n],\mathcal{Y}}\}_{n\geq 1}$ are disjoint and since there is no a priori way to (topologically) compare channels in $\DMC_{[n],\mathcal{Y}}$ with channels in $\DMC_{[n'],\mathcal{Y}}$ for $n\neq n'$, the ``most natural" topology that we can define on $\DMC_{\ast,\mathcal{Y}}$ is the disjoint union topology $\mathcal{T}_{s,\ast,\mathcal{Y}}:=\displaystyle\bigoplus_{n\geq 1}\mathcal{T}_{[n],\mathcal{Y}}$. Clearly, the space $(\DMC_{\ast,\mathcal{Y}},\mathcal{T}_{s,\ast,\mathcal{Y}})$ is disconnected. Moreover, $\mathcal{T}_{s,\ast,\mathcal{Y}}$ is metrizable because it is the disjoint union of metrizable spaces. It is also $\sigma$-compact because it is the union of countably many compact spaces.

We added the subscript $s$ to emphasize the fact that $\mathcal{T}_{s,\ast,\mathcal{Y}}$ is a strong topology (remember that the disjoint union topology is the \emph{finest} topology that makes the canonical injections continuous).

\begin{mydef}
We define the strong topology $\mathcal{T}_{s,\ast,\mathcal{Y}}^{(i)}$ on $\DMC_{\ast,\mathcal{Y}}^{(i)}$ as the quotient topology $\mathcal{T}_{s,\ast,\mathcal{Y}}/R_{\ast,\mathcal{Y}}^{(i)}$.

We call open and closed sets in $(\DMC_{\ast,\mathcal{Y}}^{(i)},\mathcal{T}_{s,\ast,\mathcal{Y}}^{(i)})$ as strongly open and strongly closed sets respectively.
\end{mydef}

Let $\Proj:\DMC_{\ast,\mathcal{Y}}\rightarrow\DMC_{\ast,\mathcal{Y}}^{(i)}$ be the projection onto the $R_{\ast,\mathcal{Y}}^{(i)}$-equivalence classes, and for every $n\geq 1$ let $\Proj_n:\DMC_{[n],\mathcal{Y}}\rightarrow\DMC_{[n],\mathcal{Y}}^{(i)}$ be the projection onto the $R_{[n],\mathcal{Y}}^{(i)}$-equivalence classes. Due to the identifications that we made in Section \ref{secInEquivSpace}, we have $\Proj(W)=\Proj_n(W)$ for every $W\in\DMC_{[n],\mathcal{Y}}$. Therefore, for every $U\subset \DMC_{\ast,\mathcal{Y}}^{(i)}$, we have
$$\textstyle\Proj^{-1}(U)={\displaystyle\coprod_{n\geq 1}}\Proj_n^{-1}(U\cap\DMC_{[n],\mathcal{Y}}^{(i)}).$$
Hence,
\begin{align*}
\textstyle U\in\mathcal{T}_{s,\ast,\mathcal{Y}}^{(i)}\;\;&\stackrel{(a)}{\Leftrightarrow}\;\;\textstyle\Proj^{-1}(U)\in\mathcal{T}_{s,\ast,\mathcal{Y}}\\
&\textstyle\stackrel{(b)}{\Leftrightarrow}\;\;\Proj^{-1}(U)\cap\DMC_{[n],\mathcal{Y}} \in\mathcal{T}_{[n],\mathcal{Y}},\;\;\forall n\geq 1\\
&\textstyle\Leftrightarrow\;\;\left({\displaystyle\coprod_{n'\geq 1}}\Proj_{n'}^{-1}(U\cap\DMC_{[n'],\mathcal{Y}}^{(i)})\right)\cap\DMC_{[n],\mathcal{Y}} \in\mathcal{T}_{[n],\mathcal{Y}},\;\;\forall n\geq 1\\
&\textstyle\Leftrightarrow\;\;\Proj_n^{-1}(U\cap\DMC_{[n],\mathcal{Y}}^{(i)}) \in\mathcal{T}_{[n],\mathcal{Y}},\;\;\forall n\geq 1\\
&\textstyle\stackrel{(c)}{\Leftrightarrow}\;\;U\cap\DMC_{[n],\mathcal{Y}}^{(i)} \in\mathcal{T}_{[n],\mathcal{Y}}^{(i)},\;\;\forall n\geq 1,
\end{align*}
where (a) and (c) follows from the properties of the quotient topology, and (b) follows from the properties of the disjoint union topology.

We conclude that $U\subset \DMC_{\ast,\mathcal{Y}}^{(i)}$ is strongly open in $\DMC_{\ast,\mathcal{Y}}^{(i)}$ if and only if $U\cap \DMC_{[n],\mathcal{Y}}^{(i)}$ is open in $\DMC_{[n],\mathcal{Y}}^{(i)}$ for every $n\geq 1$. This shows that the topology on $\DMC_{[n],\mathcal{Y}}^{(i)}$ that is inherited from $(\DMC_{\ast,\mathcal{Y}}^{(i)},\mathcal{T}_{s,\ast,\mathcal{Y}}^{(i)})$ is exactly $\mathcal{T}_{[n],\mathcal{Y}}^{(i)}$. Therefore, $\mathcal{T}_{s,\ast,\mathcal{Y}}^{(i)}$ is a natural topology. On the other hand, if $\mathcal{T}$ is an arbitrary natural topology and $U\in\mathcal{T}$, then $U\cap\DMC_{[n],\mathcal{Y}}^{(i)}$ is open in $\DMC_{[n],\mathcal{Y}}^{(i)}$ for every $n\geq 1$, so $U\in\mathcal{T}_{s,\ast,\mathcal{Y}}^{(i)}$. We conclude that $\mathcal{T}_{s,\ast,\mathcal{Y}}^{(i)}$ is the finest natural topology.

\vspace*{3mm}

We can also characterize the strongly closed subsets of $\DMC_{\ast,\mathcal{Y}}^{(i)}$ in terms of the closed sets of the $\DMC_{[n],\mathcal{Y}}^{(i)}$ spaces:
\begin{align*}
\textstyle F\;\text{is strongly closed in}\;\DMC_{\ast,\mathcal{Y}}^{(i)}\;\;&\Leftrightarrow\;\;\textstyle\DMC_{\ast,\mathcal{Y}}^{(i)}\setminus F\;\text{is strongly open in}\;\textstyle\DMC_{\ast,\mathcal{Y}}^{(i)}\\
&\textstyle\Leftrightarrow\;\;\left(\DMC_{\ast,\mathcal{Y}}^{(i)}\setminus F\right)\cap\DMC_{[n],\mathcal{Y}}^{(i)}\;\text{is open in}\;\DMC_{[n],\mathcal{Y}}^{(i)},\;\;\forall n\geq 1\\
&\textstyle\Leftrightarrow\;\;\DMC_{[n],\mathcal{Y}}^{(i)}\setminus \left(F\cap\DMC_{[n],\mathcal{Y}}^{(i)}\right)\;\text{is open in}\;\DMC_{[n],\mathcal{Y}}^{(i)},\;\;\forall n\geq 1\\
&\textstyle\Leftrightarrow\;\;F\cap\DMC_{[n],\mathcal{Y}}^{(i)}\;\text{is closed in}\;\DMC_{[n],\mathcal{Y}}^{(i)},\;\;\forall n\geq 1.
\end{align*}

Since $\DMC_{[n],\mathcal{Y}}^{(i)}$ is metrizable for every $n\geq 1$, it is also normal. We can use this fact to prove that the strong topology on $\DMC_{\ast,\mathcal{Y}}^{(i)}$ is normal:

\begin{mylem}
\label{lemDMCXiNorm}
$(\DMC_{\ast,\mathcal{Y}}^{(i)},\mathcal{T}_{s,\ast,\mathcal{Y}}^{(i)})$ is normal.
\end{mylem}
\begin{proof}
See Appendix \ref{appDMCXiNorm}.
\end{proof}

\vspace*{3mm}

The following theorem shows that the strong topology satisfies many desirable properties.

\begin{mythe}
\label{theDMCXi}
$(\DMC_{\ast,\mathcal{Y}}^{(i)},\mathcal{T}_{s,\ast,\mathcal{Y}}^{(i)})$ is a compactly generated, sequential and $T_4$ space.
\end{mythe}
\begin{proof}
Since $(\DMC_{\ast,\mathcal{Y}},\mathcal{T}_{s,\ast,\mathcal{Y}})$ is metrizable, it is sequential. Therefore, $(\DMC_{\ast,\mathcal{Y}}^{(i)},\mathcal{T}_{s,\ast,\mathcal{Y}}^{(i)})$, which is the quotient of a sequential space, is sequential.

Let us now show that $\DMC_{\ast,\mathcal{Y}}^{(i)}$ is $T_4$. Fix $\hat{W}\in\DMC_{\ast,\mathcal{Y}}^{(i)}$. For every $n\geq 1$, we have $\{\hat{W}\}\cap \DMC_{[n],\mathcal{Y}}^{(i)}$ is either $\o$ or $\{\hat{W}\}$ depending on whether $\hat{W}\in \DMC_{[n],\mathcal{Y}}^{(i)}$ or not. Since $\DMC_{[n],\mathcal{Y}}^{(i)}$ is metrizable, it is $T_1$ and so singletons are closed in $\DMC_{[n],\mathcal{Y}}^{(i)}$. We conclude that in all cases, $\{\hat{W}\}\cap \DMC_{[n],\mathcal{Y}}^{(i)}$ is closed in $\DMC_{[n],\mathcal{Y}}^{(i)}$ for every $n\geq 1$. Therefore, $\{\hat{W}\}$ is strongly closed in $\DMC_{\ast,\mathcal{Y}}^{(i)}$. This shows that $(\DMC_{\ast,\mathcal{Y}}^{(i)},\mathcal{T}_{s,\ast,\mathcal{Y}}^{(i)})$ is $T_1$. On the other hand, Lemma \ref{lemDMCXiNorm} shows that $(\DMC_{\ast,\mathcal{Y}}^{(i)},\mathcal{T}_{s,\ast,\mathcal{Y}}^{(i)})$ is normal. This means that $(\DMC_{\ast,\mathcal{Y}}^{(i)},\mathcal{T}_{s,\ast,\mathcal{Y}}^{(i)})$ is $T_4$, which implies that it is Hausdorff.

Now since $(\DMC_{\ast,\mathcal{Y}},\mathcal{T}_{s,\ast,\mathcal{Y}})$ is metrizable, it is compactly generated. On the other hand, the quotient space $(\DMC_{\ast,\mathcal{Y}}^{(i)},\mathcal{T}_{s,\ast,\mathcal{Y}}^{(i)})$ was shown to be Hausdorff. We conclude that $(\DMC_{\ast,\mathcal{Y}}^{(i)},\mathcal{T}_{s,\ast,\mathcal{Y}}^{(i)})$ is compactly generated.
\end{proof}

\begin{mycor}
If $|\mathcal{Y}|\geq 3$, $(\DMC_{\ast,\mathcal{Y}}^{(i)},\mathcal{T}_{s,\ast,\mathcal{Y}}^{(i)})$ is not locally compact anywhere.
\end{mycor}
\begin{proof}
Since $\mathcal{T}_{s,\ast,\mathcal{Y}}^{(i)}$ is a natural Hausdorff topology, Proposition \ref{propNaturalNotLocallyCompactIn} implies that $\mathcal{T}_{s,\ast,\mathcal{Y}}^{(i)}$ is not locally compact anywhere.
\end{proof}

\vspace*{3mm}

As in the case of the space of equivalent channels \cite{RajDMCTop}, the space $(\DMC_{\ast,\mathcal{Y}}^{(i)},\mathcal{T}_{s,\ast,\mathcal{Y}}^{(i)})$ fails to be first-countable (and hence it is not metrizable) when $|\mathcal{Y}|\geq 3$. This is one manifestation of the strength of the topology $\mathcal{T}_{s,\ast,\mathcal{Y}}^{(i)}$. In order to show that $(\DMC_{\ast,\mathcal{Y}}^{(i)},\mathcal{T}_{s,\ast,\mathcal{Y}}^{(i)})$ is not first-countable, we need to characterize the converging sequences in $(\DMC_{\ast,\mathcal{Y}}^{(i)},\mathcal{T}_{s,\ast,\mathcal{Y}}^{(i)})$.

A sequence $(\hat{W}_n)_{n\geq 1}$ in $\DMC_{\ast,\mathcal{Y}}^{(i)}$ is said to be \emph{rank-bounded} if $\irank(\hat{W}_n)$ is bounded. $(\hat{W}_n)_{n\geq 1}$ is \emph{rank-unbounded} if it is not bounded.

The following proposition shows that every rank-unbounded sequence does not converge in $(\DMC_{\ast,\mathcal{Y}}^{(i)},\mathcal{T}_{s,\ast,\mathcal{Y}}^{(i)})$.
\begin{myprop}
\label{propCharacConvSeqIn}
A sequence $(\hat{W}_n)_{n\geq 0}$ converges in $(\DMC_{\ast,\mathcal{Y}}^{(i)},\mathcal{T}_{s,\ast,\mathcal{Y}}^{(i)})$ if and only if there exists $m\geq 1$ such that $\hat{W}_n\in \DMC_{[m],\mathcal{Y}}^{(i)}$ for every $n\geq 0$, and $(\hat{W}_n)_{n\geq 0}$ converges in $(\DMC_{[m],\mathcal{Y}}^{(i)},\mathcal{T}_{[m],\mathcal{Y}}^{(i)})$.
\end{myprop}
\begin{proof}
Assume that a sequence $(\hat{W}_n)_{n\geq 0}$ in $\DMC_{\ast,\mathcal{Y}}^{(i)}$ is rank-unbounded. This cannot happen unless $|\mathcal{Y}|\geq 3$. In order to show that $(\hat{W}_n)_{n\geq 0}$ that does not converge, it is sufficient to show that there exists a subsequence of $(\hat{W}_n)_{n\geq 0}$ which does not converge.

Let $(\hat{W}_{n_k})_{k\geq 0}$ be any subsequence of $(\hat{W}_n)_{n\geq 0}$ where the input-rank strictly increases, i.e., $\irank(W_{n_k})<\irank(W_{n_{k'}})$ for every $0\leq k<k'$. We will show that $(\hat{W}_{n_k})_{k\geq 0}$ does not converge.

Assume to the contrary that $(\hat{W}_{n_k})_{k\geq 0}$ converges to $\hat{W}\in\DMC_{\ast,\mathcal{Y}}^{(i)}$. Define the set $$A=\{\hat{W}_{n_k}:\;k\geq 0\}\setminus\hat{W}.$$ For every $m\geq 1$, the set $A\cap \DMC_{[m],\mathcal{Y}}^{(i)}$ contains finitely many points. This means that $A\cap \DMC_{[m],\mathcal{Y}}^{(i)}$ is a finite union of singletons (which are closed in $\DMC_{[m],\mathcal{Y}}^{(i)}$), hence $A\cap \DMC_{[m],\mathcal{Y}}^{(i)}$ is closed in $\DMC_{[m],\mathcal{Y}}^{(i)}$ for every $m\geq 1$. Therefore $A$ is closed in $(\DMC_{\ast,\mathcal{Y}}^{(i)},\mathcal{T}_{s,\ast,\mathcal{Y}}^{(i)})$.

Now define $U=\DMC_{\ast,\mathcal{Y}}^{(i)}\setminus A$. Since $A$ is strongly closed, $U$ is strongly open. Moreover, $U$ contains $\hat{W}$, so $U$ is a neighborhood of $\hat{W}$. Therefore, there exists $k_0\geq 0$ such that  $\hat{W}_{n_k}\in U$ for every $k\geq k_0$. Now since the input-rank of $(\hat{W}_{n_k})_{k\geq 0}$  strictly increases, we can find $k\geq k_0$ such that $\irank(\hat{W}_{n_k})>\irank(\hat{W})$. This means that $\hat{W}_{n_k}\neq\hat{W}$ and so $\hat{W}_{n_k}\in A$. Therefore, $\hat{W}_{n_k}\notin U$ which is a contradiction.

We conclude that every converging sequence in $(\DMC_{\ast,\mathcal{Y}}^{(i)},\mathcal{T}_{s,\ast,\mathcal{Y}}^{(i)})$ must be rank-bounded.

Now let $(\hat{W}_n)_{n\geq 0}$ be a rank-bounded sequence in $\DMC_{\ast,\mathcal{Y}}^{(i)}$, i.e., there exists $m\geq 1$ such that  $\hat{W}_n\in \DMC_{[m],\mathcal{Y}}^{(i)}$ for every $n\geq 0$. If $(\hat{W}_n)_{n\geq 0}$ converges in $(\DMC_{\ast,\mathcal{Y}}^{(i)},\mathcal{T}_{s,\ast,\mathcal{Y}}^{(i)})$ then it converges in $\DMC_{[m],\mathcal{Y}}^{(i)}$ since $\DMC_{[m],\mathcal{Y}}^{(i)}$ is strongly closed.

Conversely, assume that $(\hat{W}_n)_{n\geq 0}$ converges in $(\DMC_{[m],\mathcal{Y}}^{(i)},\mathcal{T}_{[m],\mathcal{Y}}^{(i)})$ to $\hat{W}\in \DMC_{[m],\mathcal{Y}}^{(i)}$. Let $O$ be any neighborhood of $\hat{W}$ in $(\DMC_{\ast,\mathcal{Y}}^{(i)},\mathcal{T}_{s,\ast,\mathcal{Y}}^{(i)})$. There exists a strongly open set $U$ such that $\hat{W}\in U\subset O$. Since $U\cap \DMC_{[m],\mathcal{Y}}^{(i)}$ is open in $(\DMC_{[m],\mathcal{Y}}^{(i)},\mathcal{T}_{[m],\mathcal{Y}}^{(i)})$, there exists $n_0>0$ such that $\hat{W}_n\in  U\cap \DMC_{[m],\mathcal{Y}}^{(i)}$ for every $n\geq n_0$. This implies that $\hat{W}_n\in  O$ for every $n\geq n_0$. Therefore $(\hat{W}_n)_{n\geq 0}$ converges to $\hat{W}$ in $(\DMC_{\ast,\mathcal{Y}}^{(i)},\mathcal{T}_{s,\ast,\mathcal{Y}}^{(i)})$.
\end{proof}

\begin{mycor}
\label{corNotFirstCountableIn}
If $|\mathcal{Y}|\geq 3$, $(\DMC_{\ast,\mathcal{Y}}^{(i)},\mathcal{T}_{s,\ast,\mathcal{Y}}^{(i)})$ is not first-countable anywhere, i.e., for every $\hat{W}\in\DMC_{\ast,\mathcal{Y}}^{(i)}$, there is no countable neighborhood basis of $\hat{W}$.
\end{mycor}
\begin{proof}
Fix $\hat{W}\in \DMC_{\ast,\mathcal{Y}}^{(i)}$ and assume to the contrary that $\hat{W}$ admits a countable neighborhood basis $\{O_n\}_{n\geq 1}$ in $(\DMC_{\ast,\mathcal{Y}}^{(i)},\mathcal{T}_{s,\ast,\mathcal{Y}}^{(i)})$. For every $n\geq 1$, let $U_n'$ be a strongly open set such that $\hat{W}\in U_n'\subset O_n$. Define $\displaystyle U_n=\bigcap_{i=1}^n U_n'$. $U_n$ is strongly open because it is the intersection of finitely many strongly open sets. Moreover, $U_n\subset O_m$ for every $n\geq m$.

For every $n\geq 1$, Proposition \ref{propNaturallyOpenUnboundedIn} implies that $U_n$ (which is non-empty and strongly open) is rank-unbounded, so it cannot be contained in $\DMC_{[n],\mathcal{Y}}^{(i)}$. Hence there exists $\hat{W}_n\in U_n$ such that $\hat{W}_n\notin \DMC_{[n],\mathcal{Y}}^{(i)}$.

Since $\hat{W}_n\notin \DMC_{[n],\mathcal{Y}}^{(i)}$, we have $\irank(\hat{W}_n)>n$ for every $n\geq 1$. Therefore, $(\hat{W}_n)_{n\geq 1}$ is rank-unbounded. Proposition \ref{propCharacConvSeqIn} implies that $(\hat{W}_n)_{n\geq 1}$ does not converge in $(\DMC_{\ast,\mathcal{Y}}^{(i)},\mathcal{T}_{s,\ast,\mathcal{Y}}^{(i)})$.

Now let $O$ be a neighborhood of $\hat{W}$ in $(\DMC_{\ast,\mathcal{Y}}^{(i)},\mathcal{T}_{s,\ast,\mathcal{Y}}^{(i)})$. Since $\{O_n\}_{n\geq 1}$ is a neighborhood basis for $\hat{W}$, there exists $n_0\geq 1$ such that $O_{n_0}\subset O$. For every $n\geq n_0$, we have $\hat{W}_n\in U_n\subset O_{n_0}$. This means that $(\hat{W}_n)_{n\geq 1}$ converges to $\hat{W}$ in $(\DMC_{\ast,\mathcal{Y}}^{(i)},\mathcal{T}_{s,\ast,\mathcal{Y}}^{(i)})$ which is a contradiction. Therefore, $\hat{W}$ does not admit a countable neighborhood basis in $(\DMC_{\ast,\mathcal{Y}}^{(i)},\mathcal{T}_{s,\ast,\mathcal{Y}}^{(i)})$.
\end{proof}

\subsection{Compact subspaces of $(\DMC_{\ast,\mathcal{Y}}^{(i)},\mathcal{T}_{s,\ast,\mathcal{Y}}^{(i)})$}

It is well known that a compact subset of $\mathbb{R}$ is compact if and only if it is closed and bounded. The following proposition shows that a similar statement holds for $(\DMC_{\ast,\mathcal{Y}}^{(i)},\mathcal{T}_{s,\ast,\mathcal{Y}}^{(i)})$.

\begin{myprop}
\label{propCharacCompactDMCXis}
A subspace of $(\DMC_{\ast,\mathcal{Y}}^{(i)},\mathcal{T}_{s,\ast,\mathcal{Y}}^{(i)})$ is compact if and only if it is rank-bounded and strongly closed.
\end{myprop}
\begin{proof}
If $|\mathcal{Y}|=1$, $\DMC_{\ast,\mathcal{Y}}^{(i)}=\DMC_{[1],\mathcal{Y}}^{(i)}$ consists of only one point, hence all subsets of $\DMC_{\ast,\mathcal{Y}}^{(i)}$ are rank-bounded, compact and strongly closed.

If $|\mathcal{Y}|=2$, $\DMC_{\ast,\mathcal{Y}}^{(i)}=\DMC_{[2],\mathcal{Y}}^{(i)}$ and $\mathcal{T}_{s,\ast,\mathcal{Y}}^{(i)}=\mathcal{T}_{[2],\mathcal{Y}}^{(i)}$, hence all subsets of $\DMC_{\ast,\mathcal{Y}}^{(i)}$ are rank-bounded. But $\DMC_{[2],\mathcal{Y}}^{(i)}$ is compact and Hausdorff. Therefore, a subset of $\DMC_{\ast,\mathcal{Y}}^{(i)}$ is compact if and only if it is closed in $\mathcal{T}_{[2],\mathcal{Y}}^{(i)}=\mathcal{T}_{s,\ast,\mathcal{Y}}^{(i)}$.

Assume now that $|\mathcal{Y}|\geq 3$.
Let $A$ be a subspace of $(\DMC_{\ast,\mathcal{Y}}^{(i)},\mathcal{T}_{s,\ast,\mathcal{Y}}^{(i)})$. If $A$ is rank-bounded and strongly closed, then there exists $n\geq 1$ such that $A\subset \DMC_{[n],\mathcal{Y}}^{(i)}$. Since $A$ is strongly closed, then $A=A\cap \DMC_{[n],\mathcal{Y}}^{(i)}$ is closed in $\DMC_{[n],\mathcal{Y}}^{(i)}$ which is compact. Therefore, $A$ is compact.

Now let $A$ be a compact subspace of $(\DMC_{\ast,\mathcal{Y}}^{(i)},\mathcal{T}_{s,\ast,\mathcal{Y}}^{(i)})$. Since $(\DMC_{\ast,\mathcal{Y}}^{(i)},\mathcal{T}_{s,\ast,\mathcal{Y}}^{(i)})$ is Hausdorff, $A$ is strongly closed. It remains to show that $A$ is rank-bounded.

Assume to the contrary that $A$ is rank-unbounded. We can construct a sequence $(\hat{W}_n)_{n\geq 0}$ in $A$ where the input-rank is strictly increasing, i.e., $\irank(\hat{W}_n)<\irank(\hat{W}_{n'})$ for every $0\leq n<n'$. Since the input-rank of $(\hat{W}_n)_{n\geq 0}$ is strictly increasing, every subsequence of $(\hat{W}_n)_{n\geq 0}$ is rank-unbounded. Proposition \ref{propCharacConvSeqIn} implies that every subsequence of $(\hat{W}_n)_{n\geq 0}$ does not converge in $(\DMC_{\ast,\mathcal{Y}}^{(i)},\mathcal{T}_{s,\ast,\mathcal{Y}}^{(i)})$. On the other hand, we have:
\begin{itemize}
\item $A$ is countably compact because it is compact.
\item Since $A$ is strongly closed and since $(\DMC_{\ast,\mathcal{Y}}^{(i)},\mathcal{T}_{s,\ast,\mathcal{Y}}^{(i)})$ is a sequential space, $A$ is sequential.
\item $A$ is Hausdorff because $(\DMC_{\ast,\mathcal{Y}}^{(i)},\mathcal{T}_{s,\ast,\mathcal{Y}}^{(i)})$ is Hausdorff.
\end{itemize}
Now since every countably compact sequential Hausdorff space is sequentially compact \cite{SequentialSpace}, $A$ must be sequentially compact. Therefore, $(\hat{W}_n)_{n\geq 0}$ has a converging subsequence which is a contradiction. We conclude that $A$ must be rank-bounded.
\end{proof}

\section{The similarity metric on the space of input-equivalent channels}

We define the \emph{similarity metric} on $\DMC_{\ast,\mathcal{Y}}^{(i)}$ as follows:
\begin{align*}
d_{\ast,\mathcal{Y}}^{(i)}(\hat{W}_1,\hat{W}_2)&=\min_{R\in \mathcal{R}(\conv(\hat{W}_1),\conv(\hat{W}_2))}\max_{(P_1,P_2)\in R} \|P_1-P_2\|_{TV}\\
&=\frac{1}{2}\min_{R\in \mathcal{R}(\conv(\hat{W}_1),\conv(\hat{W}_2))}\max_{(P_1,P_2)\in R} \sum_{y\in\mathcal{Y}} |P_1(y)-P_2(y)|.
\end{align*}

Let $\mathcal{T}_{\ast,\mathcal{Y}}^{(i)}$ be the metric topology on $\DMC_{\ast,\mathcal{Y}}^{(i)}$ that is induced by $d_{\ast,\mathcal{Y}}^{(i)}$.  We call $\mathcal{T}_{\ast,\mathcal{Y}}^{(i)}$ the \emph{similarity topology} on $\DMC_{\ast,\mathcal{Y}}^{(i)}$. 

Clearly, $\mathcal{T}_{\ast,\mathcal{Y}}^{(i)}$ is natural because the restriction of $d_{\ast,\mathcal{Y}}^{(i)}$ on $\DMC_{[n],\mathcal{Y}}^{(i)}$ is exactly $d_{[n],\mathcal{Y}}^{(i)}$, and the topology induced by $d_{[n],\mathcal{Y}}^{(i)}$ is $\mathcal{T}_{[n],\mathcal{Y}}^{(i)}$ (Theorem \ref{theDMCXYi}).

\section{Continuity of channel parameters and operations}

\subsection{Channel parameters}

The \emph{capacity} of a channel $W\in\DMC_{\mathcal{X},\mathcal{Y}}$ is denoted as $C(W)$.

An \emph{$(n,M)$-encoder} on the alphabet $\mathcal{X}$  is a mapping $\mathcal{E}:\mathcal{M}\rightarrow\mathcal{X}^n$ such that $|\mathcal{M}|=M$. The set $\mathcal{M}$ is the \emph{message set} of $\mathcal{E}$, $n$ is the \emph{blocklength} of $\mathcal{E}$, $M$ is the \emph{size} of $\mathcal{E}$, and $\frac{1}{n}\log M$ is the \emph{rate} of $\mathcal{E}$ (measured in nats). The \emph{error probability of the ML decoder for the encoder $\mathcal{E}$ when it is used for a channel $W\in\DMC_{\mathcal{X},\mathcal{Y}}$} is given by:
$$P_{e,\mathcal{E}}(W)=1-\frac{1}{M}\sum_{y_1^n\in\mathcal{Y}^n} \max_{m\in\mathcal{M}}\left\{\prod_{i=1}^n W(y_i|\mathcal{E}_i(m))\right\},$$
where $(\mathcal{E}_1(m),\ldots,\mathcal{E}_n(m))=\mathcal{E}(m)$.

The \emph{optimal error probability of $(n,M)$-encoders for a channel $W$} is given by:
$$P_{e,n,M}(W)=\min_{\substack{\mathcal{E}\;\text{is an}\\(n,M)\text{-encoder}}}P_{e,\mathcal{E}}(W).$$

Since input-degradedness is a particular case of the Shannon ordering \cite{ShannonDegrad}, we can easily see that if $W$ and $W'$ are input-equivalent, then $C(W)=C(W')$ and $P_{e,n,M}(W)=P_{e,n,M}(W')$ for every $n\geq 1$ and every $M\geq 1$. Therefore, for every $\hat{W}\in\DMC_{\ast,\mathcal{Y}}^{(i)}$, we can define $C(\hat{W}):=C(W')$ for any $W'\in\hat{W}$. We can define $P_{e,n,M}(\hat{W})$ similarly. Moreover, due to Proposition \ref{propInputDegradOperational}, we can also define $P_{e,\mathcal{D}}(\hat{W})$ for any decoder $\mathcal{D}$ on the output alphabet $\mathcal{Y}$.

\begin{myprop}
\label{propContParamDMCXYiStr}
Let $\mathcal{X}$ and $\mathcal{Y}$ be two finite sets. We have:
\begin{itemize}
\item $C:\DMC_{\mathcal{X},\mathcal{Y}}^{(i)}\rightarrow \mathbb{R}^+$ is continuous on $(\DMC_{\mathcal{X},\mathcal{Y}}^{(i)},\mathcal{T}_{\mathcal{X},\mathcal{Y}}^{(i)})$.
\item For every $n\geq 1$ and every $M\geq 1$, the mapping $P_{e,n,M}:\DMC_{\mathcal{X},\mathcal{Y}}^{(i)}\rightarrow [0,1]$ is continuous on $(\DMC_{\mathcal{X},\mathcal{Y}}^{(i)},\mathcal{T}_{\mathcal{X},\mathcal{Y}}^{(i)})$.
\item For every decoder $\mathcal{D}$ on $\mathcal{Y}$, the mapping $P_{e,\mathcal{D}}:\DMC_{\mathcal{X},\mathcal{Y}}^{(i)}\rightarrow [0,1]$ is continuous on $(\DMC_{\mathcal{X},\mathcal{Y}}^{(i)},\mathcal{T}_{\mathcal{X},\mathcal{Y}}^{(i)})$.
\end{itemize}
\end{myprop}
\begin{proof}
Since $C:\DMC_{\mathcal{X},\mathcal{Y}}\rightarrow\mathbb{R}^+$ is continuous, and since $C(W)$ depends only on the $R_{\mathcal{X},\mathcal{Y}}^{(i)}$, Lemma \ref{lemQuotientFunction} implies that $C:\DMC_{\mathcal{X},\mathcal{Y}}^{(i)}\rightarrow \mathbb{R}^+$ is continuous on $(\DMC_{\mathcal{X},\mathcal{Y}}^{(i)},\mathcal{T}_{\mathcal{X},\mathcal{Y}}^{(i)})$. We can show the continuity of $P_{e,n,M}$ and $P_{e,\mathcal{D}}$ on $(\DMC_{\mathcal{X},\mathcal{Y}}^{(i)},\mathcal{T}_{\mathcal{X},\mathcal{Y}}^{(i)})$ similarly.
\end{proof}

\vspace*{3mm}

The following lemma provides a way to check whether a mapping defined on $(\DMC_{\ast,\mathcal{Y}}^{(i)},\mathcal{T}_{s,\ast,\mathcal{Y}}^{(i)})$ is continuous:

\begin{mylem}
\label{lemContinuityForStrongTopIn}
Let $(S,\mathcal{V})$ be an arbitrary topological space. A mapping $f:\DMC_{\ast,\mathcal{Y}}^{(i)}\rightarrow S$ is continuous on $(\DMC_{\ast,\mathcal{Y}}^{(i)},\mathcal{T}_{s,\ast,\mathcal{Y}}^{(i)})$ if and only if it is continuous on $(\DMC_{[n],\mathcal{Y}}^{(i)},\mathcal{T}_{[n],\mathcal{Y}}^{(i)})$ for every $n\geq 1$.
\end{mylem}
\begin{proof}
\begin{align*}
\textstyle f\;\text{is continuous on}\;(\DMC_{\ast,\mathcal{Y}}^{(i)},\mathcal{T}_{s,\ast,\mathcal{Y}}^{(i)})\;\;&\textstyle\Leftrightarrow\;\; f^{-1}(V)\in \mathcal{T}_{s,\ast,\mathcal{Y}}^{(i)},\;\;\forall V\in\mathcal{V}\\
&\textstyle\Leftrightarrow\;\; f^{-1}(V)\cap \DMC_{[n],\mathcal{Y}}^{(i)} \in \mathcal{T}_{[n],\mathcal{Y}}^{(i)},\;\;\forall n\geq 1,\; \forall V\in\mathcal{V}\\
&\textstyle\Leftrightarrow\;\; f\;\text{is continuous on}\; (\DMC_{[n],\mathcal{Y}}^{(i)},\mathcal{T}_{[n],\mathcal{Y}}^{(i)}),\;\;\forall n\geq 1.
\end{align*}
\end{proof}

\begin{myprop}
\label{propContParamDMCXiStr}
Let $\mathcal{Y}$ be a finite set. We have:
\begin{itemize}
\item $C:\DMC_{\ast,\mathcal{Y}}^{(i)}\rightarrow \mathbb{R}^+$ is continuous on $(\DMC_{\ast,\mathcal{Y}}^{(i)},\mathcal{T}_{s,\ast,\mathcal{Y}}^{(i)})$.
\item For every $n\geq 1$ and every $M\geq 1$, the mapping $P_{e,n,M}:\DMC_{\ast,\mathcal{Y}}^{(i)}\rightarrow [0,1]$ is continuous on $(\DMC_{\ast,\mathcal{Y}}^{(i)},\mathcal{T}_{s,\ast,\mathcal{Y}}^{(i)})$.
\item For every decoder $\mathcal{D}$ on $\mathcal{Y}$, the mapping $P_{e,\mathcal{D}}:\DMC_{\ast,\mathcal{Y}}^{(i)}\rightarrow [0,1]$ is continuous on $(\DMC_{\ast,\mathcal{Y}}^{(i)},\mathcal{T}_{s,\ast,\mathcal{Y}}^{(i)})$.
\end{itemize}
\end{myprop}
\begin{proof}
The proposition follows from Proposition \ref{propContParamDMCXYiStr} and Lemma \ref{lemContinuityForStrongTopIn}.
\end{proof}

\subsection{Channel operations}

\label{subsecChanOperIn}

For every two channels $W_1\in \DMC_{\mathcal{X}_1,\mathcal{Y}_1}$ and $W_2\in \DMC_{\mathcal{X}_2,\mathcal{Y}_2}$, define the \emph{channel sum} $W_1 \oplus W_2\in \DMC_{\mathcal{X}_1\coprod\mathcal{X}_2,\mathcal{Y}_1\coprod\mathcal{Y}_2}$ of $W_1$ and $W_2$ as:
$$(W_1\oplus W_2)(y,i|x,j)=\begin{cases}W_i(y|x)\quad&\text{if}\;i=j,\\
0&\text{otherwise},\end{cases}$$
where $\mathcal{X}_1\coprod\mathcal{X}_2=(\mathcal{X}_1\times\{1\})\cup(\mathcal{X}_2\times\{2\})$ is the disjoint union of $\mathcal{X}_1$ and $\mathcal{X}_2$. $W_1\oplus W_2$ arises when the transmitter has two channels $W_1$ and $W_2$ at his disposal and he can use exactly one of them at each channel use.

We define the \emph{channel product} $W_1\otimes W_2\in \DMC_{\mathcal{X}_1\times\mathcal{X}_2,\mathcal{Y}_1\times\mathcal{Y}_2}$ of $W_1$ and $W_2$ as:
$$(W_1\otimes W_2)(y_1,y_2|x_1,x_2)=W_1(y_1|x_1)W_2(y_2|x_2).$$
$W_1\otimes W_2$ arises when the transmitter has two channels $W_1$ and $W_2$ at his disposal and he uses both of them at each channel use. Channel sums and products were first introduced by Shannon in \cite{ChannelSumProduct}.

Channel sums and products can be ``quotiented" by the input-equivalence relation. We just need to realize that the input-equivalence class of the resulting channel depends only on the input-equivalence classes of the channels that were used in the operation. Let us illustrate this in the case of channel sums:

Let $W_1,W_1'\in \DMC_{\mathcal{X}_1,\mathcal{Y}_1}$ and $W_2,W_2'\in \DMC_{\mathcal{X}_2,\mathcal{Y}_2}$ and assume that $W_1$ is input-degraded from $W_1'$ and $W_2$ is input-degraded from $W_2'$. There exists $V_1'\in\DMC_{\mathcal{X}_1,\mathcal{X}_1}$ and $V_2'\in\DMC_{\mathcal{X}_2,\mathcal{X}_2}$ such that $W_1=W_1'\circ V_1'$ and $W_2=W_2'\circ V_2'$. It is easy to see that $W_1\oplus W_2=(W_1'\oplus W_2')\circ (V_1'\oplus V_2')$, which shows that $W_1\oplus W_2$ is input-degraded from $W_1'\oplus W_2'$.

Therefore, if $W_1$ is input-equivalent to $W_1'$ and $W_2$ is input-equivalent to $W_2'$, then $W_1\oplus W_2$ is input-equivalent to $W_1'\oplus W_2'$. This allows us to define the channel sum for every $\hat{W}_1\in\DMC_{\mathcal{X}_1,\mathcal{Y}_1}^{(i)}$ and every $\overline{W}_2\in\DMC_{\mathcal{X}_2,\mathcal{Y}_2}^{(i)}$ as $\hat{W}_1\oplus \overline{W}_2 = \widetilde{W_1'\oplus W_2'}\in \DMC_{\mathcal{X}_1\coprod\mathcal{X}_2,\mathcal{Y}_1\coprod\mathcal{Y}_2}^{(i)}$ for any $W_1'\in\hat{W}_1$ and any $W_2'\in\overline{W}_2$, where $\widetilde{W_1'\oplus W_2'}$ is the $R_{\mathcal{X}_1\coprod\mathcal{X}_2,\mathcal{Y}_1\coprod\mathcal{Y}_2}^{(i)}$-equivalence class of $W_1'\oplus W_2'$. We can define the product on the quotient spaces similarly.

\begin{myprop}
\label{propContOperDMCXYi}
We have:
\begin{itemize}
\item The mapping $(\hat{W}_1,\overline{W}_2)\rightarrow \hat{W}_1\oplus \overline{W}_2$ from $\DMC_{\mathcal{X}_1,\mathcal{Y}_1}^{(i)}\times \DMC_{\mathcal{X}_2,\mathcal{Y}_2}^{(i)}$ to $\DMC_{\mathcal{X}_1\coprod \mathcal{X}_2,\mathcal{Y}_1\coprod\mathcal{Y}_2}^{(i)}$ is continuous.
\item The mapping $(\hat{W}_1,\overline{W}_2)\rightarrow \hat{W}_1\otimes \overline{W}_2$ from $\DMC_{\mathcal{X}_1,\mathcal{Y}_1}^{(i)}\times \DMC_{\mathcal{X}_2,\mathcal{Y}_2}^{(i)}$ to $\DMC_{\mathcal{X}_1\times\mathcal{X}_2,\mathcal{Y}_1\times\mathcal{Y}_2}^{(i)}$ is continuous.
\end{itemize}
\end{myprop}
\begin{proof}
We only prove the continuity of the channel sum because the proof for the channel product is similar.

Let $\Proj:\DMC_{\mathcal{X}_1\coprod\mathcal{X}_2,\mathcal{Y}_1\coprod\mathcal{Y}_2}\rightarrow \DMC_{\mathcal{X}_1\coprod\mathcal{X}_2,\mathcal{Y}_1\coprod\mathcal{Y}_2}^{(i)}$ be the projection onto the $R_{\mathcal{X}_1\coprod\mathcal{X}_2,\mathcal{Y}_1\coprod\mathcal{Y}_2}^{(i)}$-equivalence classes. Define the mapping $f:\DMC_{\mathcal{X}_1,\mathcal{Y}_1}\times \DMC_{\mathcal{X}_2,\mathcal{Y}_2}\rightarrow \DMC_{\mathcal{X}_1\coprod\mathcal{X}_2,\mathcal{Y}_1\coprod\mathcal{Y}_2}^{(i)}$ as $f(W_1,W_2)=\Proj(W_1\oplus W_2)$. Clearly, $f$ is continuous.

Now define the equivalence relation $R$ on $\DMC_{\mathcal{X}_1,\mathcal{Y}_1}\times \DMC_{\mathcal{X}_2,\mathcal{Y}_2}$ as:
$$(W_1,W_2)R(W_1',W_2')\;\;\Leftrightarrow\;\; W_1 R_{\mathcal{X}_1,\mathcal{Y}_1}^{(i)}W_1'\;\text{and}\;W_2 R_{\mathcal{X}_2,\mathcal{Y}_2}^{(i)}W_2'.$$
The discussion before the proposition shows that $f(W_1,W_2)=\Proj(W_1\oplus W_2)$ depends only on the $R$-equivalence class of $(W_1,W_2)$. Lemma \ref{lemQuotientFunction} now shows that the transcendent map of $f$ defined on $(\DMC_{\mathcal{X}_1,\mathcal{Y}_1}\times \DMC_{\mathcal{X}_2,\mathcal{Y}_2})/R$ is continuous.

Notice that $(\DMC_{\mathcal{X}_1,\mathcal{Y}_1}\times \DMC_{\mathcal{X}_2,\mathcal{Y}_2})/R$ can be identified with $\DMC_{\mathcal{X}_1,\mathcal{Y}_1}^{(i)}\times \DMC_{\mathcal{X}_2,\mathcal{Y}_2}^{(i)}$. Therefore, we can define $f$ on $\DMC_{\mathcal{X}_1,\mathcal{Y}_1}^{(i)}\times \DMC_{\mathcal{X}_2,\mathcal{Y}_2}^{(i)}$ through this identification. Moreover, since $\DMC_{\mathcal{X}_1,\mathcal{Y}_1}$ and $\DMC_{\mathcal{X}_2,\mathcal{Y}_2}^{(i)}$ are locally compact and Hausdorff, Corollary \ref{corQuotientProd} implies that the canonical bijection between $(\DMC_{\mathcal{X}_1,\mathcal{Y}_1}\times \DMC_{\mathcal{X}_2,\mathcal{Y}_2})/R$ and $\DMC_{\mathcal{X}_1,\mathcal{Y}_1}^{(i)}\times \DMC_{\mathcal{X}_2,\mathcal{Y}_2}^{(i)}$ is a homeomorphism. 

Now since the mapping $f$ on $\DMC_{\mathcal{X}_1,\mathcal{Y}_1}^{(i)}\times \DMC_{\mathcal{X}_2,\mathcal{Y}_2}^{(i)}$ is just the channel sum, we conclude that the mapping $(\hat{W}_1,\overline{W}_2)\rightarrow \hat{W}_1\oplus \overline{W}_2$ from $\DMC_{\mathcal{X}_1,\mathcal{Y}_1}^{(i)}\times \DMC_{\mathcal{X}_2,\mathcal{Y}_2}^{(i)}$ to $\DMC_{\mathcal{X}_1\coprod\mathcal{X}_2,\mathcal{Y}_1\coprod\mathcal{Y}_2}^{(i)}$ is continuous.
\end{proof}

\begin{myprop}
\label{propContOperDMCXiStr}
Assume that all spaces of input-equivalent channels are endowed with the strong topology. We have:
\begin{itemize}
\item The mapping $(\hat{W}_1,\overline{W}_2)\rightarrow \hat{W}_1\oplus \overline{W}_2$ from $\DMC_{\ast,\mathcal{Y}_1}^{(i)}\times \DMC_{\mathcal{X}_2,\mathcal{Y}_2}^{(i)}$ to $\DMC_{\ast,\mathcal{Y}_1\coprod\mathcal{Y}_2}^{(i)}$ is continuous.
\item The mapping $(\hat{W}_1,\overline{W}_2)\rightarrow \hat{W}_1\otimes \overline{W}_2$ from $\DMC_{\ast,\mathcal{Y}_1}^{(i)}\times \DMC_{\mathcal{X}_2,\mathcal{Y}_2}^{(i)}$ to $\DMC_{\ast,\mathcal{Y}_1\times\mathcal{Y}_2}^{(i)}$ is continuous.
\end{itemize}
\end{myprop}
\begin{proof}
We only prove the continuity of the channel sum because the proof of the continuity of the channel product is similar.

Due to the distributivity of the product with respect to disjoint unions, we have:
$$\textstyle\DMC_{\ast,\mathcal{Y}_1}\times\DMC_{\mathcal{X}_2,\mathcal{Y}_2}={\displaystyle\coprod_{n\geq1}}(\DMC_{[n],\mathcal{Y}_1}\times\DMC_{\mathcal{X}_2,\mathcal{Y}_2}),$$
and
$$\textstyle\mathcal{T}_{s,\ast,\mathcal{Y}_1}\otimes\mathcal{T}_{\mathcal{X}_2,\mathcal{Y}_2}={\displaystyle\bigoplus_{n\geq1}}\left(\mathcal{T}_{[n],\mathcal{Y}_1}\otimes\mathcal{T}_{\mathcal{X}_2,\mathcal{Y}_2}\right).$$

Therefore, the space $\DMC_{\ast,\mathcal{Y}_1}\times\DMC_{\mathcal{X}_2,\mathcal{Y}_2}$ is the topological disjoint union of the spaces $(\DMC_{[n],\mathcal{Y}_1}\times\DMC_{\mathcal{X}_2,\mathcal{Y}_2})_{n\geq 1}$.

For every $n\geq 1$, let $\Proj_n$ be the projection onto the $R_{[n]\coprod\mathcal{X}_2,\mathcal{Y}_1\coprod\mathcal{Y}_2}^{(i)}$-equivalence classes and let $i_n$ be the canonical injection from $\DMC_{[n]\coprod\mathcal{X}_2,\mathcal{Y}_1\coprod\mathcal{Y}_2}^{(i)}$ to $\DMC_{\ast,\mathcal{Y}_1\coprod\mathcal{Y}_2}^{(i)}$. 

Define the mapping $f: \DMC_{\ast,\mathcal{Y}_1}\times\DMC_{\mathcal{X}_2,\mathcal{Y}_2}\rightarrow \DMC_{\ast,\mathcal{Y}_1\coprod \mathcal{Y}_2}^{(i)}$ as $$\textstyle f(W_1,W_2)=i_{n}(\Proj_{n}(W_1\oplus W_2))=\hat{W}_1\oplus\overline{W}_2,$$
where $n$ is the unique integer satisfying $W_1\in \DMC_{[n],\mathcal{Y}_1}$. $\hat{W}_1$ and $\overline{W}_2$ are the $R_{[n],\mathcal{Y}_1}^{(i)}$ and $R_{\mathcal{X}_2,\mathcal{Y}_2}^{(i)}$-equivalence classes of $W_1$ and $W_2$ respectively.

Clearly, the mapping $f$ is continuous on $\DMC_{[n],\mathcal{Y}_1}\times\DMC_{\mathcal{X}_2,\mathcal{Y}_2}$ for every $n\geq 1$. Therefore, $f$ is continuous on $(\DMC_{\ast,\mathcal{Y}_1}\times\DMC_{\mathcal{X}_2,\mathcal{Y}_2},\mathcal{T}_{s,\ast,\mathcal{Y}_1}\otimes\mathcal{T}_{\mathcal{X}_2,\mathcal{Y}_2})$.

Let $R$ be the equivalence relation defined on $\DMC_{\ast,\mathcal{Y}_1}\times\DMC_{\mathcal{X}_2,\mathcal{Y}_2}$ as follows: $(W_1,W_2)R(W_1',W_2')$ if and only if $W_1 R_{\ast,\mathcal{Y}}^{(i)} W_1'$ and $W_2 R_{\mathcal{X}_2,\mathcal{Y}_2}^{(i)} W_2'$.

Since $f(W_1,W_2)$ depends only on the $R$-equivalence class of $(W_1,W_2)$, Lemma \ref{lemQuotientFunction} implies that the transcendent mapping of $f$ is continuous on $(\DMC_{\ast,\mathcal{Y}_1}\times\DMC_{\mathcal{X}_2,\mathcal{Y}_2})/R$.

Since $(\DMC_{\ast,\mathcal{Y}_1},\mathcal{T}_{s,\ast,\mathcal{Y}_1})$ and $\DMC_{\mathcal{X}_2,\mathcal{Y}_2}^{(i)}=\DMC_{\mathcal{X}_2,\mathcal{Y}_2}/R_{\mathcal{X}_2,\mathcal{Y}_2}^{(i)}$ are Hausdorff and locally compact, Corollary \ref{corQuotientProd} implies that the canonical bijection from $\DMC_{\ast,\mathcal{Y}_1}^{(i)}\times \DMC_{\mathcal{X}_2,\mathcal{Y}_2}^{(i)}$ to $(\DMC_{\ast,\mathcal{Y}_1}\times \DMC_{\mathcal{X}_2,\mathcal{Y}_2})/R$ is a homeomorphism. We conclude that the channel sum is continuous on $(\DMC_{\ast,\mathcal{Y}_1}^{(i)}\times \DMC_{\mathcal{X}_2,\mathcal{Y}_2}^{(i)},\mathcal{T}_{s,\ast,\mathcal{Y}_1}^{(i)}\otimes\mathcal{T}_{\mathcal{X}_2,\mathcal{Y}}^{(i)})$.
\end{proof}

\vspace*{3mm}

The reader might be wondering why the channel sum and the channel product were not shown to be continuous on the whole space $\DMC_{\ast,\mathcal{Y}_1}^{(i)}\times \DMC_{\ast,\mathcal{Y}_2}^{(i)}$ instead of the smaller space $\DMC_{\ast,\mathcal{Y}_1}^{(i)}\times \DMC_{\mathcal{X}_2,\mathcal{Y}_2}^{(i)}$. The reason is because we cannot apply Corollary \ref{corQuotientProd} to $\DMC_{\ast,\mathcal{Y}_1}\times \DMC_{\ast,\mathcal{Y}_2}$ and $\DMC_{\ast,\mathcal{Y}_1}^{(i)}\times \DMC_{\ast,\mathcal{Y}_2}^{(i)}$ since neither $\DMC_{\ast,\mathcal{Y}_1}^{(i)}$ nor $\DMC_{\ast,\mathcal{Y}_2}^{(i)}$ is locally compact when $|\mathcal{Y}_1|,|\mathcal{Y}_2|\geq 3$ (under the strong topology).

As in the case of the space of equivalent channels \cite{RajContTop}, one potential method to show the continuity of the channel sum on $(\DMC_{\ast,\mathcal{Y}_1}^{(i)}\times\DMC_{\ast,\mathcal{Y}_2}^{(i)},\mathcal{T}_{s,\ast,\mathcal{Y}_1}^{(i)}\otimes \mathcal{T}_{s,\ast,\mathcal{Y}_2}^{(i)})$ is as follows: let $R$ be the equivalence relation on $\DMC_{\ast,\mathcal{Y}_1}\times\DMC_{\ast,\mathcal{Y}_2}$ defined as $(W_1,W_2)R(W_1',W_2')$ if and only if $W_1 R_{\ast,\mathcal{Y}_1}^{(i)}W_1'$ and $W_2 R_{\ast,\mathcal{Y}_2}^{(i)}W_2'$. We can identify $(\DMC_{\ast,\mathcal{Y}_1}\times\DMC_{\ast,\mathcal{Y}_2})/R$ with $\DMC_{\ast,\mathcal{Y}_1}^{(i)}\times\DMC_{\ast,\mathcal{Y}_2}^{(i)}$ through the canonical bijection. Using Lemma \ref{lemQuotientFunction}, it is easy to see that the mapping $(\hat{W}_1,\overline{W}_2)\rightarrow \hat{W}_1\oplus\overline{W}_2$ is continuous from $\big(\DMC_{\ast,\mathcal{Y}_1}^{(i)}\times\DMC_{\ast,\mathcal{Y}_2}^{(i)}, (\mathcal{T}_{s,\ast,\mathcal{Y}_1}\otimes \mathcal{T}_{s,\ast,\mathcal{Y}_2})/R\big)$ to $(\DMC_{\ast,\mathcal{Y}_1\coprod\mathcal{Y}_2}^{(i)},\mathcal{T}_{s,\ast,\mathcal{Y}_1\coprod\mathcal{Y}_2}^{(i)})$.

It was shown in \cite{CompactlyGenerated} that the topology $(\mathcal{T}_{s,\ast,\mathcal{Y}_1}\otimes \mathcal{T}_{s,\ast,\mathcal{Y}_2})/R$ is homeomorphic to $\kappa(\mathcal{T}_{s,\ast,\mathcal{Y}_1}^{(i)}\otimes \mathcal{T}_{s,\ast,\mathcal{Y}_2}^{(i)})$ through the canonical bijection, where $\kappa(\mathcal{T}_{s,\ast,\mathcal{Y}_1}^{(i)}\otimes \mathcal{T}_{s,\ast,\mathcal{Y}_2}^{(i)})$ is the coarsest topology that is both compactly generated and finer than $\mathcal{T}_{s,\ast,\mathcal{Y}_1}^{(i)}\otimes \mathcal{T}_{s,\ast,\mathcal{Y}_2}^{(i)}$. Therefore, the mapping $(\hat{W}_1,\overline{W}_2)\rightarrow \hat{W}_1\oplus\overline{W}_2$ is continuous on $\big(\DMC_{\ast,\mathcal{Y}_1}^{(i)}\times\DMC_{\ast,\mathcal{Y}_2}^{(i)}, \kappa(\mathcal{T}_{s,\ast,\mathcal{Y}_1}^{(i)}\otimes \mathcal{T}_{s,\ast,\mathcal{Y}_2}^{(i)})\big)$. This means that if $\mathcal{T}_{s,\ast,\mathcal{Y}_1}^{(i)}\otimes \mathcal{T}_{s,\ast,\mathcal{Y}_2}^{(i)}$ is compactly generated, we will have $\mathcal{T}_{s,\ast,\mathcal{Y}_1}^{(i)}\otimes \mathcal{T}_{s,\ast,\mathcal{Y}_2}^{(i)}=\kappa(\mathcal{T}_{s,\ast,\mathcal{Y}_1}^{(i)}\otimes \mathcal{T}_{s,\ast,\mathcal{Y}_2}^{(i)})$ and so the channel sum will be continuous on $(\DMC_{\ast,\mathcal{Y}_1}^{(i)}\times\DMC_{\ast,\mathcal{Y}_2}^{(i)}, \mathcal{T}_{s,\ast,\mathcal{Y}_1}^{(i)}\otimes \mathcal{T}_{s,\ast,\mathcal{Y}_2}^{(i)})$. Note that although $\mathcal{T}_{s,\ast,\mathcal{Y}_1}^{(i)}$ and $\mathcal{T}_{s,\ast,\mathcal{Y}_2}^{(i)}$ are compactly generated, their product $\mathcal{T}_{s,\ast,\mathcal{Y}_1}^{(i)}\otimes \mathcal{T}_{s,\ast,\mathcal{Y}_2}^{(i)}$ might not be compactly generated.

\begin{myprop}
\label{propChanOperFormIn}
Let $\mathcal{Y}_1$ and $\mathcal{Y}_2$ be two finite set. Let $\hat{W}_1\in\DMC_{\ast,\mathcal{Y}_1}^{(i)}$ and $\overline{W}_2\in\DMC_{\ast,\mathcal{Y}_2}^{(i)}$. We have:
$$\conv(\hat{W}_1\oplus\overline{W}_2)= \bigcup_{0\leq \lambda\leq 1}\Big((1-\lambda)\phi_{1\#}(\conv(\hat{W}_1))+\lambda\phi_{2\#}(\conv(\hat{W}_2))\Big),$$
where $\phi_{1\#}$ and $\phi_{2\#}$ are the push-forwards by the canonical injections from $\mathcal{Y}_1$ and $\mathcal{Y}_2$ to $\mathcal{Y}_1\coprod\mathcal{Y}_2$ respectively. On the other hand,
$$\conv(\hat{W}_1\otimes\overline{W}_2)=\conv\Big(\conv(\hat{W}_1)\otimes \conv(\hat{W}_2)\Big).$$
\end{myprop}
\begin{proof}
See Appendix \ref{appChanOperFormIn}.
\end{proof}

\begin{myprop}
\label{propContOperDMCXiSimilar}
Assume that all spaces of input-equivalent channels are endowed with the similarity topology. We have:
\begin{itemize}
\item The mapping $(\hat{W}_1,\overline{W}_2)\rightarrow \hat{W}_1\oplus \overline{W}_2$ from $\DMC_{\ast,\mathcal{Y}_1}^{(i)}\times \DMC_{\ast,\mathcal{Y}_2}^{(i)}$ to $\DMC_{\ast,\mathcal{Y}_1\coprod\mathcal{Y}_2}^{(i)}$ is continuous.
\item The mapping $(\hat{W}_1,\overline{W}_2)\rightarrow \hat{W}_1\otimes \overline{W}_2$ from $\DMC_{\ast,\mathcal{Y}_1}^{(i)}\times \DMC_{\ast,\mathcal{Y}_2}^{(i)}$ to $\DMC_{\ast,\mathcal{Y}_1\times\mathcal{Y}_2}^{(i)}$ is continuous.
\end{itemize}
\end{myprop}
\begin{proof}
See Appendix \ref{appContOperDMCXiSimilar}.
\end{proof}

\section{The natural Borel $\sigma$-algebra on $\DMC_{\ast,\mathcal{Y}}^{(i)}$}

Let $\mathcal{T}$ be a Hausdorff natural topology on $\DMC_{\ast,\mathcal{Y}}^{(i)}$. Since $\mathcal{T}_{s,\ast,\mathcal{Y}}^{(i)}$ is the finest natural topology, we have $\mathcal{T}\subset \mathcal{T}_{s,\ast,\mathcal{Y}}^{(i)}$. Therefore, $\mathcal{B}(\mathcal{T})\subset \mathcal{B}(\mathcal{T}_{s,\ast,\mathcal{Y}}^{(i)})$, where $\mathcal{B}(\mathcal{T})$ and $\mathcal{B}(\mathcal{T}_{s,\ast,\mathcal{Y}}^{(i)})$ are the Borel $\sigma$-algebras of $\mathcal{T}$ and $\mathcal{T}_{s,\ast,\mathcal{Y}}^{(i)}$ respectively.

On the other hand, for every $U\in \mathcal{T}_{s,\ast,\mathcal{Y}}^{(i)}$ and every $n\geq 1$, we have $U\cap\DMC_{[n],\mathcal{Y}}^{(i)}\in\mathcal{T}_{[n],\mathcal{Y}}^{(i)}$. But $\mathcal{T}$ is a natural topology, so there must exist $U_n\in \mathcal{T}$ such that $U_n\cap\DMC_{[n],\mathcal{Y}}^{(i)}=U\cap\DMC_{[n],\mathcal{Y}}^{(i)}$. Since $U_n\in\mathcal{T}$, we have $U_n\in\mathcal{B}(\mathcal{T})$. Moreover, $\DMC_{[n],\mathcal{Y}}^{(i)}$ is $\mathcal{T}$-closed (because it is compact and $\mathcal{T}$ is Hausdorff). Therefore, $\DMC_{[n],\mathcal{Y}}^{(i)}\in\mathcal{B}(\mathcal{T})$. This implies that $U\cap\DMC_{[n],\mathcal{Y}}^{(i)}=U_n\cap\DMC_{[n],\mathcal{Y}}^{(i)}\in\mathcal{B}(T)$, hence
$$U=\bigcup_{n\geq 1}(U\cap {\DMC}_{[n],\mathcal{Y}}^{(i)})\in\mathcal{B}(T).$$
Since this is true for every $U\in \mathcal{T}_{s,\ast,\mathcal{Y}}^{(i)}$, we have $\mathcal{T}_{s,\ast,\mathcal{Y}}^{(i)}\subset\mathcal{B}(\mathcal{T})$ which implies that $\mathcal{B}(\mathcal{T}_{s,\ast,\mathcal{Y}}^{(i)})\subset\mathcal{B}(\mathcal{T})$. We conclude that all Hausdorff natural topologies on $\DMC_{\ast,\mathcal{Y}}^{(i)}$ have the same $\sigma$-algebra.  This $\sigma$-algebra deserves to be called the \emph{natural Borel $\sigma$-algebra} on $\DMC_{\ast,\mathcal{Y}}^{(i)}$.

Note that for every $n\geq 1$, the inclusion mapping $i_n: \DMC_{[n],\mathcal{Y}}^{(i)}\rightarrow\DMC_{\ast,\mathcal{Y}}^{(i)}$ is continuous from $(\DMC_{[n],\mathcal{Y}}^{(i)},\mathcal{T}_{[n],\mathcal{Y}}^{(i)})$ to $(\DMC_{\ast,\mathcal{Y}}^{(i)},\mathcal{T}_{s,\ast,\mathcal{Y}}^{(i)})$, hence it is measurable. Therefore, for every $B\in \mathcal{B}(\mathcal{T}_{s,\ast,\mathcal{Y}}^{(i)})$, we have $i_n^{-1}(B)=B\cap\DMC_{[n],\mathcal{Y}}^{(i)}\in\mathcal{B}(\mathcal{T}_{[n],\mathcal{Y}}^{(i)})$. In the following, we show a converse for this statement.

Fix $n\geq 1$ and let $U\in\mathcal{T}_{[n],\mathcal{Y}}^{(i)}$. There exists $U'\in\mathcal{T}_{s,\ast,\mathcal{Y}}^{(i)}$ such that $U=U'\cap\DMC_{[n],\mathcal{Y}}^{(i)}$. Since $U'$ and $\DMC_{[n],\mathcal{Y}}^{(i)}$ are respectively open and closed in the topology $\mathcal{T}_{s,\ast,\mathcal{Y}}^{(i)}$, they are both in its Borel $\sigma$-algebra. Therefore, $U=U'\cap\DMC_{[n],\mathcal{Y}}^{(i)}\in\mathcal{B}(\mathcal{T}_{s,\ast,\mathcal{Y}}^{(i)})$ for every $U\in\mathcal{T}_{[n],\mathcal{Y}}^{(i)}$. This means that $\mathcal{T}_{[n],\mathcal{Y}}^{(i)}\subset \mathcal{B}(\mathcal{T}_{s,\ast,\mathcal{Y}}^{(i)})$ and $\mathcal{B}(\mathcal{T}_{[n],\mathcal{Y}}^{(i)})\subset \mathcal{B}(\mathcal{T}_{s,\ast,\mathcal{Y}}^{(i)})$ for every $n\geq 1$.

Assume now that $A\subset\DMC_{\ast,\mathcal{Y}}^{(i)}$ satisfies $A\cap\DMC_{[n],\mathcal{Y}}^{(i)} \in\mathcal{B}(\mathcal{T}_{[n],\mathcal{Y}}^{(i)})$ for every $n\geq 1$. This implies that $A\cap\DMC_{[n],\mathcal{Y}}^{(i)} \in\mathcal{B}(\mathcal{T}_{s,\ast,\mathcal{Y}}^{(i)})$ for every $n\geq 1$, hence
$$A=\bigcup_{n\geq 1}(A\cap{\DMC}_{[n],\mathcal{Y}}^{(i)}) \in\mathcal{B}(\mathcal{T}_{s,\ast,\mathcal{Y}}^{(i)}).$$
We conclude that a subset $A$ of $\DMC_{\ast,\mathcal{Y}}^{(i)}$ is in the natural Borel $\sigma$-algebra if and only if $A\cap\DMC_{[n],\mathcal{Y}}^{(i)} \in\mathcal{B}(\mathcal{T}_{[n],\mathcal{Y}}^{(i)})$ for every $n\geq 1$.

\section{Conclusion}

Since $\mathcal{T}_{\ast,\mathcal{Y}}^{(i)}$ is a natural topology, it is not completely metrizable because of Corollary \ref{corNaturalNotCompleteIn}. Therefore, the metric space $(\DMC_{\ast,\mathcal{Y}}^{(i)},d_{\ast,\mathcal{Y}}^{(i)})$ is not complete. An interesting question to ask is: what does the completion of $(\DMC_{\ast,\mathcal{Y}}^{(i)},d_{\ast,\mathcal{Y}}^{(i)})$ represent? Does it represent the space of all input-equivalent channels with output alphabet $\mathcal{Y}$ and arbitrary input alphabet (with arbitrary cardinality)?

Many other interesting questions remain open: Are all natural topologies Hausdorff? Can we find more topological properties that are common for all natural topologies? Is there a coarsest natural topology? Is there a natural topology that is coarser than the similarity one?

The continuity of the channel parameters $C$, $P_{e,n,M}$ and $P_{e,\mathcal{D}}$ on $\mathcal{T}_{\ast,\mathcal{Y}}^{(i)}$ is an open problem. Also, the continuity of the channel sum and the channel product on the whole product space $(\DMC_{\ast,\mathcal{Y}_1}^{(i)}\times \DMC_{\ast,\mathcal{Y}_2}^{(i)},\mathcal{T}_{s,\ast,\mathcal{Y}_1}^{(i)}\otimes \mathcal{T}_{s,\ast,\mathcal{Y}_2}^{(i)})$ remains an open problem. As we explained in Section \ref{subsecChanOperIn}, it is sufficient to prove that the product topology $\mathcal{T}_{s,\ast,\mathcal{Y}_1}^{(i)}\otimes \mathcal{T}_{s,\ast,\mathcal{Y}_2}^{(i)}$ is compactly generated.

In \cite{RaginskyShannon}, Raginsky introduced the Shannon deficiency. We can define the \emph{input-deficiency} similarly. Like the Shannon deficiency, the input deficiency compares a particular channel with the input-equivalence class of another channel. The input deficiency is not a metric distance between input-equivalence classes of channels.

\section*{Acknowledgment}

I would like to thank Emre Telatar for helpful discussions. I am also grateful to Maxim Raginsky for informing me about the work of Blackwell on statistical experiments.

\appendices

\section{Proof of Proposition \ref{propInteriorEmptyDMCXni}}
\label{appInteriorEmptyDMCXni}

If $|\mathcal{Y}|=1$, then $\Delta_{\mathcal{Y}}$ contains only one point and so $|\CE(W)|=1$ for every $W\in \DMC_{[n],\mathcal{Y}}$ and every $n\geq 1$. Therefore, $\DMC_{[n],\mathcal{Y}}^{(i)}=\DMC_{[1],\mathcal{Y}}^{(i)}$ for every $n\geq 1$.

If $|\mathcal{Y}|=2$, then $\Delta_{\mathcal{Y}}$ is a one dimensional segment. Therefore, there are at most two convex-extreme points for any finite subset of $\Delta_{\mathcal{Y}}$. This means that $|\CE(W)|\leq 2$ for every $W\in \DMC_{[n],\mathcal{Y}}$ and every $n\geq 2$. Therefore, $\DMC_{[n],\mathcal{Y}}^{(i)}=\DMC_{[2],\mathcal{Y}}^{(i)}$ for every $n\geq 2$.

Now assume that $|\mathcal{Y}|\geq 3$. Let $\hat{U}$ be an arbitrary non-empty open subset of $(\DMC_{[m],\mathcal{Y}}^{(i)},\mathcal{T}_{[m],\mathcal{Y}}^{(i)})$ and let $\Proj$ be the projection onto the $R_{[m],\mathcal{Y}}^{(i)}$-equivalence classes. $\Proj^{-1}(\hat{U})$ is open in the metric space $(\DMC_{[m],\mathcal{Y}},d_{[m],\mathcal{Y}})$. Let $\hat{W}\in\hat{U}$ and define $r=\irank(\hat{W})$. Let $P_1,\ldots,P_r\in\Delta_{\mathcal{Y}}$ be such that $\CE(\hat{W})=\{P_1,\ldots,P_r\}$. Define the channel $W\in\DMC_{[m],\mathcal{Y}}$ as follows:
$$W(y|i)=\begin{cases}P_i(y)\quad&\text{if}\;1\leq i<r,\\ P_r(y)\quad&\text{if}\;r\leq i\leq m.\end{cases} $$
Clearly $\CE(W)=\CE(\hat{W})$ and so $W\in\hat{W}$ which implies that $W\in \Proj^{-1}(\hat{U})$. Since $\Proj^{-1}(\hat{U})$ is open in the metric space $(\DMC_{[m],\mathcal{Y}},d_{[m],\mathcal{Y}})$, there exists $\epsilon>0$ such that $\Proj^{-1}(\hat{U})$ contains the open ball of center $W$ and radius $\epsilon$.

We will show that there exists $W'\in \DMC_{[m],\mathcal{Y}}$ such that $\irank(W')=m>n$ and $d_{[m],\mathcal{Y}}(W,W')<\epsilon$. If $r=\irank(W)=m$, take $W'=W$.

Assume that $r=\irank(W)<m$. Since $|\mathcal{Y}|\geq 3$, the dimension of $\Delta_{\mathcal{Y}}$ is at least 2. Therefore, we can find $P_{r+1}\in\Delta_{\mathcal{Y}}$ such that $\|P_r-P_{r+1}\|_{TV}<\epsilon$ and $\CE(\{P_1,\ldots,P_{r+1}\})=\{P_1,\ldots,P_{r+1}\}$. By repeating this procedure $m-r$ times, we obtain $P_{r+1},\ldots,P_m\in\Delta_{\mathcal{Y}}$ such that $\|P_r-P_{i}\|_{TV}<\epsilon$ for every $r+1\leq i\leq m$, and $\CE(\{P_1,\ldots,P_{m}\})=\{P_1,\ldots,P_{m}\}$. Define the channel $W'\in\Delta_{[m],\mathcal{Y}}$ as:
$$W'(y|i)=P_i(y).$$
We have $\CE(W')=\CE(\{P_1,\ldots,P_m\})=\{P_1,\ldots,P_m\}$. Therefore, $\irank(W')=m$. Moreover,
$$d_{[m],\mathcal{Y}}(W,W')=\max_{1\leq i\leq m}\|W_i-W_i'\|_{TV}=\max_{r+1\leq i\leq m} \|P_r-P_i\|_{TV}<\epsilon.$$
This means that $W'\in\Proj^{-1}(\hat{U})$ and $W'$ is not input-equivalent to any channel in $\DMC_{[n],\mathcal{Y}}$ (see Proposition \ref{propCharacInputEquiv}). Therefore, $\Proj(W')\in \hat{U}$ and $\Proj(W')\notin\DMC_{[n],\mathcal{Y}}^{(i)}$ because $W'$ is not input-equivalent to any channel in $\DMC_{[n],\mathcal{Y}}$. This shows that every non-empty open subset of $\DMC_{[m],\mathcal{Y}}^{(i)}$ is not contained in $\DMC_{[n],\mathcal{Y}}^{(i)}$. We conclude that the interior of $\DMC_{[n],\mathcal{Y}}^{(i)}$ in $\DMC_{[m],\mathcal{Y}}^{(i)}$ is empty.

\section{Proof of Lemma \ref{lemDMCXiNorm}}
\label{appDMCXiNorm}

Define $\DMC_{[0],\mathcal{Y}}^{(i)}=\o$, which is strongly closed in $\DMC_{\ast,\mathcal{Y}}^{(i)}$.

Let $A$ and $B$ be two disjoint strongly closed subsets of $\DMC_{\ast,\mathcal{Y}}^{(i)}$. For every $n\geq 0$, let $A_n=A\cap \DMC_{[n],\mathcal{Y}}^{(i)}$ and $B_n=B\cap \DMC_{[n],\mathcal{Y}}^{(i)}$. Since $A$ and $B$ are strongly closed in $\DMC_{\ast,\mathcal{Y}}^{(i)}$, $A_n$ and $B_n$ are closed in $\DMC_{[n],\mathcal{Y}}^{(i)}$. Moreover, $A_n\cap B_n\subset A\cap B=\o$.

Construct the sequences $(U_n)_{n\geq 0},(U_n')_{n\geq 0},(K_n)_{n\geq 0}$ and $(K_n')_{n\geq 0}$ recursively as follows:

$U_0=U_0'=K_0=K_0'=\o\subset\DMC_{[0],\mathcal{Y}}^{(i)}$. Since $A_0=B_0=\o$, we have $A_0\subset U_0\subset K_0$ and $B_0\subset U_0'\subset K_0'$. Moreover, $U_0$ and $U_0'$ are open in $\DMC_{[0],\mathcal{Y}}^{(i)}$, $K_0$ and $K_0'$ are closed in $\DMC_{[0],\mathcal{Y}}^{(i)}$, and $K_0\cap K_0'=\o$.

Now let $n\geq 1$ and assume that we constructed $(U_j)_{0\leq j< n},(U_j')_{0\leq j< n},(K_j)_{0\leq j< n}$ and $(K_j')_{0\leq j< n}$ such that for every $0\leq j< n$, we have $A_j\subset U_j\subset K_j\subset\DMC_{[j],\mathcal{Y}}^{(i)}$, $B_j\subset U_j'\subset K_j'\subset \DMC_{[j],\mathcal{Y}}^{(i)}$, $U_j$ and $U_j'$ are open in $\DMC_{[j],\mathcal{Y}}^{(i)}$, $K_j$ and $K_j'$ are closed in $\DMC_{[j],\mathcal{Y}}^{(i)}$, and $K_j\cap K_j'=\o$. Moreover, assume that $K_j\subset U_{j+1}$ and $K_j'\subset U_{j+1}'$ for every $0\leq j<n-1$.

Let $C_n=A_n\cup K_{n-1}$ and $D_n=B_n\cup K_{n-1}'$. Since $K_{n-1}$ and $K_{n-1}'$ are closed in $\DMC_{[n-1],\mathcal{Y}}^{(i)}$ and since $\DMC_{[n-1],\mathcal{Y}}^{(i)}$ is closed in $\DMC_{[n],\mathcal{Y}}^{(i)}$, we can see that $K_{n-1}$ and $K_{n-1}'$ are closed in $\DMC_{[n],\mathcal{Y}}^{(i)}$. Therefore, $C_n$ and $D_n$ are closed in $\DMC_{[n],\mathcal{Y}}^{(i)}$. Moreover, we have
\begin{align*}
C_n\cap D_n&=(A_n\cup K_{n-1})\cap(B_n\cup K_{n-1}')\\
&=(A_n\cap B_n)\cup(A_n\cap K_{n-1}')\cup (K_{n-1}\cap B_n)\cup(K_{n-1}\cap K_{n-1}')\\
&\stackrel{(a)}{=}\textstyle \left(A_n\cap K_{n-1}'\cap \DMC_{[n-1],\mathcal{Y}}^{(i)}\right)\cup \left(K_{n-1}\cap \DMC_{[n-1],\mathcal{Y}}^{(i)}\cap B_n\right)\\
&=(A_{n-1}\cap K_{n-1}')\cup (K_{n-1}\cap B_{n-1})\subset (K_{n-1}\cap K_{n-1}')\cup (K_{n-1}\cap K_{n-1}')=\o,
\end{align*}
where (a) follows from the fact that $A_n\cap B_n=K_{n-1}\cap K_{n-1}'=\o$ and the fact that $K_{n-1}\subset \DMC_{[n-1],\mathcal{Y}}^{(i)}$ and $K_{n-1}'\subset \DMC_{[n-1],\mathcal{Y}}^{(i)}$.

Since $\DMC_{[n],\mathcal{Y}}^{(i)}$ is normal (because it is metrizable), and since $C_n$ and $D_n$ are closed disjoint subsets of $\DMC_{[n],\mathcal{Y}}^{(i)}$, there exist two sets $U_n,U_n'\subset \DMC_{[n],\mathcal{Y}}^{(i)}$ that are open in $\DMC_{[n],\mathcal{Y}}^{(i)}$ and two sets $K_n,K_n'\subset \DMC_{[n],\mathcal{Y}}^{(i)}$ that are closed in $\DMC_{[n],\mathcal{Y}}^{(i)}$ such that $C_n\subset U_n\subset K_n$, $D_n\subset U_n'\subset K_n'$ and $K_n\cap K_n'=\o$. Clearly, $A_n\subset U_n\subset K_n\subset \DMC_{[n],\mathcal{Y}}^{(i)}$, $B_n\subset U_n'\subset K_n'\subset \DMC_{[n],\mathcal{Y}}^{(i)}$, $K_{n-1}\subset U_n$ and $K_{n-1}'\subset U_n'$. This concludes the recursive construction.

Now define $\displaystyle U=\bigcup_{n\geq 0}U_n=\bigcup_{n\geq 1}U_n$ and $\displaystyle U'=\bigcup_{n\geq 0}U_n'=\bigcup_{n\geq 1}U_n'$. Since $A_n\subset U_n$ for every $n\geq 1$, we have 
\begin{align*}
\textstyle A=A\cap\DMC_{\ast,\mathcal{Y}}^{(i)}=A\cap\left({\displaystyle\bigcup_{n\geq 1}}\DMC_{[n],\mathcal{Y}}^{(i)}\right)={\displaystyle\bigcup_{n\geq 1}}\left(A\cap \DMC_{[n],\mathcal{Y}}^{(i)}\right)={\displaystyle\bigcup_{n\geq 1}} A_n\subset {\displaystyle\bigcup_{n\geq 1}} U_n =U.
\end{align*}
Moreover, for every $n\geq 1$ we have
\begin{align*}
\textstyle U\cap \DMC_{[n],\mathcal{Y}}^{(i)}=\left({\displaystyle\bigcup_{j\geq 1} U_j}\right)\cap \DMC_{[n],\mathcal{Y}}^{(i)}\stackrel{(a)}{=}\left({\displaystyle\bigcup_{j\geq n} U_j}\right)\cap \DMC_{[n],\mathcal{Y}}^{(i)}={\displaystyle\bigcup_{j\geq n} \left(U_j\cap \textstyle\DMC_{[n],\mathcal{Y}}^{(i)}\right)},
\end{align*}
where (a) follows from the fact that $U_j\subset K_j\subset U_{j+1}$ for every $j\geq 0$, which means that the sequence $(U_j)_{j\geq 1}$ is increasing.

For every $j\geq n$, we have $\DMC_{[n],\mathcal{Y}}^{(i)}\subset \DMC_{[j],\mathcal{Y}}^{(i)}$ and $U_j$ is open in $\DMC_{[j],\mathcal{Y}}^{(i)}$, hence $U_j\cap \DMC_{[n],\mathcal{Y}}^{(i)}$ is open in $\DMC_{[n],\mathcal{Y}}^{(i)}$. Therefore, $U\cap \DMC_{[n],\mathcal{Y}}^{(i)}=\displaystyle\bigcup_{j\geq n} \left(U_j\cap \textstyle\DMC_{[n],\mathcal{Y}}^{(i)}\right)$ is open in $\DMC_{[n],\mathcal{Y}}^{(i)}$. Since this is true for every $n\geq 1$, we conclude that $U$ is strongly open in $\DMC_{\ast,\mathcal{Y}}^{(i)}$.

We can show similarly that $B\subset U'$ and that $U'$ is strongly open in $\DMC_{\ast,\mathcal{Y}}^{(i)}$. Finally, we have
\begin{align*}
U\cap U'=\left(\bigcup_{n\geq 1} U_n\right)\cap \left(\bigcup_{n'\geq 1} U_{n'}'\right)=\bigcup_{n\geq 1, n'\geq 1}(U_n\cap U_{n'}')\stackrel{(a)}{=}\bigcup_{n\geq 1}(U_n\cap U_n')
&\subset\bigcup_{n\geq 1}(K_n\cap K_n')=\o,
\end{align*}
where (a) follows from the fact that for every $n\geq 1$ and every $n'\geq 1$, we have $$U_n\cap U_{n'}'\subset U_{\max\{n,n'\}}\cap U_{\max\{n,n'\}}'$$ because $(U_n)_{n\geq 1}$ and $(U_n')_{n\geq 1}$ are increasing. We conclude that $(\DMC_{\ast,\mathcal{Y}}^{(i)},\mathcal{T}_{s,\ast,\mathcal{Y}}^{(i)})$ is normal.

\section{Proof of Proposition \ref{propChanOperFormIn}}

\label{appChanOperFormIn}

Fix $W_1\in\hat{W}_1$ and $W_2\in\overline{W}_2$, and let $\mathcal{X}_1$ and $\mathcal{X}_2$ be the input alphabets of $W_1$ and $W_2$ respectively.

For every $x_1\in\mathcal{X}_1$, we have $(W_1\oplus W_2)_{x_1}=\phi_{1\#}(W_1)_{x_1}$. Similarly, for every $x_2\in\mathcal{X}_2$, we have $(W_1\oplus W_2)_{x_2}=\phi_{2\#}(W_2)_{x_2}$. Therefore,
\begin{align*}
\conv(\hat{W}_1\oplus\overline{W}_2)&=\conv\left(\left\{(W_1\oplus W_2)_x:\;x\in\mathcal{X}_1\coprod\mathcal{X}_2\right\}\right)\\
&=\conv(\{(W_1\oplus W_2)_{x_1}:\;x_1\in\mathcal{X}_1\}\cup \{(W_1\oplus W_2)_{x_2}:\;x_2\in\mathcal{X}_2\})\\
&=\conv(\{\phi_{1\#}(W_1)_{x_1}:\;x_1\in\mathcal{X}_1\}\cup \{\phi_{2\#}(W_2)_{x_2}:\;x_2\in\mathcal{X}_2\})\\
&=\bigcup_{0\leq \lambda\leq 1} \Big((1-\lambda)\conv(\{\phi_{1\#}(W_1)_{x_1}:\;x_1\in\mathcal{X}_1\})+\lambda \conv(\{\phi_{2\#}(W_2)_{x_2}:\;x_2\in\mathcal{X}_2\})\Big)\\
&=\bigcup_{0\leq \lambda\leq 1}\Big( (1-\lambda) \phi_{1\#}\big(\conv(\{(W_1)_{x_1}:\;x_1\in\mathcal{X}_1\})\big)+\lambda \phi_{2\#}\big(\conv(\{(W_2)_{x_2}:\;x_2\in\mathcal{X}_2\})\big)\Big)\\
&=\bigcup_{0\leq \lambda\leq 1}\Big( (1-\lambda) \phi_{1\#}(\conv(W_1))+\lambda \phi_{2\#}(\conv(W_2))\Big)\\
&=\bigcup_{0\leq \lambda\leq 1}\Big( (1-\lambda) \phi_{1\#}(\conv(\hat{W}_1))+\lambda \phi_{2\#}(\conv(\overline{W}_2))\Big).
\end{align*}

For every $(x_1,x_2)\in\mathcal{X}_1\times\mathcal{X}_2$, we have $(W_1\otimes W_2)_{(x_1,x_2)}=(W_1)_{x_1}\times(W_2)_{x_2}$. Therefore,
\begin{align*}
\conv(\hat{W}_1\otimes\overline{W}_2)&=\conv(\{(W_1\otimes W_2)_{(x_1,x_2)}:\;(x_1,x_2)\in\mathcal{X}_1\times\mathcal{X}_2\})\\
&=\conv(\{(W_1)_{x_1}\times(W_2)_{x_2}:\;(x_1,x_2)\in\mathcal{X}_1\times\mathcal{X}_2\})\\
&=\conv(\{(W_1)_{x_1}:\;x_1\in\mathcal{X}_1\}\otimes \{(W_2)_{x_2}:\;x_2\in\mathcal{X}_2\})\\
&=\conv\Big(\conv\big(\{(W_1)_{x_1}:\;x_1\in\mathcal{X}_1\}\big)\otimes \conv\big(\{(W_2)_{x_2}:\;x_2\in\mathcal{X}_2\}\big)\Big)\\
&=\conv\Big(\conv(W_1)\otimes \conv(W_2)\Big)=\conv\Big(\conv(\hat{W}_1)\otimes \conv(\overline{W}_2)\Big).
\end{align*}

\section{Proof of Proposition \ref{propContOperDMCXiSimilar}}
\label{appContOperDMCXiSimilar}

Fix $\hat{W}_1,\hat{W}_1'\in\DMC_{\ast,\mathcal{Y}_1}^{(i)}$ and $\overline{W}_2,\overline{W}_2'\in\DMC_{\ast,\mathcal{Y}_2}^{(i)}$. Let $R_1\in\mathcal{R}(\conv(\hat{W}_1),\conv(\hat{W}_1'))$ and $R_2\in\mathcal{R}(\conv(\overline{W}_2),\conv(\overline{W}_2'))$. Fix $0\leq \lambda\leq 1$, $(P_1,P_1')\in R_1$ and $(P_2,P_2')\in R_2$. Let $P=(1-\lambda)\phi_{1\#} P_1 + \lambda\phi_{2\#} P_2$ and $P'=(1-\lambda)\phi_{1\#} P_1' + \lambda\phi_{2\#} P_2'$, where $\phi_{1\#}$ and $\phi_{2\#}$ are the push-forwards by the canonical injections from $\mathcal{Y}_1$ and $\mathcal{Y}_2$ to $\mathcal{Y}_1\coprod\mathcal{Y}_2$ respectively. We have:
\begin{equation}
\label{eqTVIneqConvOplus}
\begin{aligned}
\|P-P'\|_{TV}&=\left\|\big((1-\lambda)\phi_{1\#} P_1 + \lambda\phi_{2\#} P_2\big) -\big((1-\lambda)\phi_{1\#} P_1' + \lambda\phi_{2\#} P_2'\big) \right\|_{TV}\\
&\leq (1-\lambda)\|\phi_{1\#}P_1 -\phi_{1\#}P_1'\|_{TV}+\lambda\|\phi_{2\#}P_2-\phi_{2\#}P_2'\|_{TV}\\
&=(1-\lambda)\|P_1 -P_1'\|_{TV}+\lambda\|P_2-P_2'\|_{TV}\\
&\leq\|P_1 -P_1'\|_{TV}+\|P_2 -P_2'\|_{TV}.
\end{aligned}
\end{equation}
Proposition \ref{propChanOperFormIn} shows that $$\conv(\hat{W}_1\oplus\overline{W}_2)= \bigcup_{0\leq \lambda\leq 1}\Big((1-\lambda)\phi_{1\#}(\conv(\hat{W}_1))+\lambda\phi_{2\#}(\conv(\hat{W}_2))\Big),$$
and
$$\conv(\hat{W}_1'\oplus\overline{W}_2')= \bigcup_{0\leq \lambda\leq 1}\Big((1-\lambda)\phi_{1\#}(\conv(\hat{W}_1'))+\lambda\phi_{2\#}(\conv(\hat{W}_2'))\Big).$$
Define $R\subset \conv(\hat{W}_1\oplus\overline{W}_2)\times \conv(\hat{W}_1'\oplus\overline{W}_2')$ as follows:
\begin{align*}
R=\Big\{\big((1-\lambda)\phi_{1\#}P_1+\lambda \phi_{2\#}P_2,(1-\lambda)\phi_{1\#}P_1'+\lambda \phi_{2\#}P_2'\big):\;0\leq \lambda\leq 1,(P_1,P_1')\in R_1,(P_2,P_2')\in R_2\Big\}.
\end{align*}
It is easy to see that $R$ is a coupling of $\conv(\hat{W}_1\oplus\overline{W}_2)$ and $\conv(\hat{W}_1'\oplus\overline{W}_2')$. We have:
\begin{align*}
d_{\ast,\mathcal{Y}_1\coprod\mathcal{Y}_2}^{(i)}(\hat{W}_1\oplus\overline{W}_2,\hat{W}_1'\oplus\overline{W}_2')&\leq \sup_{(P,P')\in R}\|P-P'\|_{TV}\\
&\stackrel{(a)}{\leq} \sup_{(P_1,P_1')\in R_1}\|P_1-P_1'\|_{TV}+\sup_{(P_2,P_2')\in R_2}\|P_2-P_2'\|_{TV},
\end{align*}
where (a) follows from \eqref{eqTVIneqConvOplus}. Since this is true for every $R_1\in\mathcal{R}(\conv(\hat{W}_1),\conv(\hat{W}_1'))$ and every $R_2\in\mathcal{R}(\conv(\hat{W}_2),\conv(\hat{W}_2'))$, we conclude that
\begin{align*}
d_{\ast,\mathcal{Y}_1\coprod\mathcal{Y}_2}^{(i)}&(\hat{W}_1\oplus\overline{W}_2,\hat{W}_1'\oplus\overline{W}_2')\\
&\leq \inf_{R_1\in\mathcal{R}(\conv(\hat{W}_1),\conv(\hat{W}_1'))} \sup_{(P_1,P_1')\in R_1}\|P_1-P_1'\|_{TV}+\inf_{R_2\in\mathcal{R}(\conv(\hat{W}_2),\conv(\hat{W}_2'))}\sup_{(P_2,P_2')\in R_2}\|P_2-P_2'\|_{TV}\\
&=d_{\ast,\mathcal{Y}_1}^{(i)}(\hat{W}_1,\hat{W}_1')+d_{\ast,\mathcal{Y}_2}^{(i)}(\hat{W}_2,\hat{W}_2').
\end{align*}
This shows that the mapping $(\hat{W}_1,\overline{W}_2)\rightarrow \hat{W}_1\oplus \overline{W}_2$ from $\DMC_{\ast,\mathcal{Y}_1}^{(i)}\times \DMC_{\ast,\mathcal{Y}_2}^{(i)}$ to $\DMC_{\ast,\mathcal{Y}_1\coprod\mathcal{Y}_2}^{(i)}$ is continuous in the similarity topology.

Fix again $R_1\in\mathcal{R}(\conv(\hat{W}_1),\conv(\hat{W}_1'))$ and $R_2\in\mathcal{R}(\conv(\overline{W}_2),\conv(\overline{W}_2'))$. Let $\lambda_1,\ldots,\lambda_k\geq 0$ be such that $\displaystyle\sum_{i=1}^k\lambda_i=1$. Let $(P_{1,1},P_{1,1}'),\ldots,(P_{1,k},P_{1,k}')\in R_1$ and $(P_{2,1},P_{2,1}'),\ldots,(P_{2,k},P_{2,k}')\in R_2$. Define $\displaystyle P=\sum_{i=1}^k\lambda_i P_{1,i}\times P_{2,i}$ and $\displaystyle P'=\sum_{i=1}^k\lambda_i P_{1,i}'\times P_{2,i}'$. We have:
\begin{equation}
\label{eqTVIneqConvOtimes}
\begin{aligned}
\|P-P'\|_{TV}&=\left\|\left(\sum_{i=1}^k\lambda_i P_{1,i}\times P_{2,i}\right) -\left(\sum_{i=1}^k\lambda_i P_{1,i}'\times P_{2,i}'\right) \right\|_{TV}\\
&\leq \sum_{i=1}^k\lambda_i \|(P_{1,i}\times P_{2,i})-(P_{1,i}'\times P_{2,i}')\|_{TV}\\
&\stackrel{(a)}{\leq} \sum_{i=1}^k\lambda_i \Big( \|P_{1,i}-P_{1,i}'\|_{TV} + \|P_{2,i}-P_{2,i}'\|_{TV}\Big)\\
&\leq \sup_{(P_1,P_1')\in R_1} \|P_1-P_1'\|_{TV} + \sup_{(P_2,P_2')\in R_2} \|P_2-P_2'\|_{TV},
\end{aligned}
\end{equation}
where (a) follows from \cite[App. B]{RajContTop}. Proposition \ref{propChanOperFormIn} shows that $$\conv(\hat{W}_1\otimes\overline{W}_2)= \conv\Big(\conv(\hat{W}_1)\otimes\conv(\overline{W}_2)\Big),$$
and
$$\conv(\hat{W}_1'\otimes\overline{W}_2')= \conv\Big(\conv(\hat{W}_1')\otimes\conv(\overline{W}_2')\Big).$$
Define $R\subset \conv(\hat{W}_1\otimes\overline{W}_2)\times \conv(\hat{W}_1'\otimes\overline{W}_2')$ as follows:
\begin{align*}
R=\Bigg\{\left(\sum_{i=1}^k\lambda_i P_{1,i}\times P_{2,i},\sum_{i=1}^k\lambda_i P_{1,i}'\times P_{2,i}'\right):\;&k\geq 1,\;\lambda_1,\ldots,\lambda_k\geq 0,\; \sum_{i=1}^k\lambda_i=1,\\&(P_{1,1},P_{1,1}'),\ldots,(P_{1,k},P_{1,k}')\in R_1,\\
&(P_{2,1},P_{2,1}'),\ldots,(P_{2,k},P_{2,k}')\in R_2\Bigg\}.
\end{align*}
It is easy to see that $R$ is a coupling of $\conv(\hat{W}_1\otimes\overline{W}_2)$ and $\conv(\hat{W}_1'\otimes\overline{W}_2')$. We have:
\begin{align*}
d_{\ast,\mathcal{Y}_1\times\mathcal{Y}_2}^{(i)}(\hat{W}_1\otimes\overline{W}_2,\hat{W}_1'\otimes\overline{W}_2')&\leq \sup_{(P,P')\in R}\|P-P'\|_{TV}\\
&\stackrel{(a)}{\leq} \sup_{(P_1,P_1')\in R_1}\|P_1-P_1'\|_{TV}+\sup_{(P_2,P_2')\in R_2}\|P_2-P_2'\|_{TV},
\end{align*}
where (a) follows from \eqref{eqTVIneqConvOtimes}. Since this is true for every $R_1\in\mathcal{R}(\conv(\hat{W}_1),\conv(\hat{W}_1'))$ and every $R_2\in\mathcal{R}(\conv(\hat{W}_2),\conv(\hat{W}_2'))$, we conclude that
\begin{align*}
d_{\ast,\mathcal{Y}_1\times\mathcal{Y}_2}^{(i)}&(\hat{W}_1\otimes\overline{W}_2,\hat{W}_1'\otimes\overline{W}_2')\\
&\leq \inf_{R_1\in\mathcal{R}(\conv(\hat{W}_1),\conv(\hat{W}_1'))} \sup_{(P_1,P_1')\in R_1}\|P_1-P_1'\|_{TV}+\inf_{R_2\in\mathcal{R}(\conv(\hat{W}_2),\conv(\hat{W}_2'))}\sup_{(P_2,P_2')\in R_2}\|P_2-P_2'\|_{TV}\\
&=d_{\ast,\mathcal{Y}_1}^{(i)}(\hat{W}_1,\hat{W}_1')+d_{\ast,\mathcal{Y}_2}^{(i)}(\hat{W}_2,\hat{W}_2').
\end{align*}
This shows that the mapping $(\hat{W}_1,\overline{W}_2)\rightarrow \hat{W}_1\otimes \overline{W}_2$ from $\DMC_{\ast,\mathcal{Y}_1}^{(i)}\times \DMC_{\ast,\mathcal{Y}_2}^{(i)}$ to $\DMC_{\ast,\mathcal{Y}_1\coprod\mathcal{Y}_2}^{(i)}$ is continuous in the similarity topology.

\bibliographystyle{IEEEtran}
\bibliography{bibliofile}
\end{document}